\documentclass[a4paper,onecolumn,11pt]{quantumarticle}
\pdfoutput=1
\usepackage[utf8]{inputenc}
\usepackage[english]{babel}
\usepackage[T1]{fontenc}
\usepackage{amsmath}
\usepackage{amssymb}
\usepackage{hyperref}
\usepackage{amsthm}
\usepackage{tikz}
\usepackage{lipsum}
\usepackage{enumerate}
\usepackage[shortlabels]{enumitem}

\newtheorem{theorem}{Theorem}
\newtheorem{definition}{Definition}
\newtheorem{lemma}[theorem]{Lemma}
\begin{document}

\title{Assessing the quality of near-term photonic quantum devices}

\author{Rawad Mezher}
\affiliation{Quandela SAS, 7 Rue Léonard de Vinci, 91300 Massy, France}
%\orcid{0000-0002-2445-2701}
\email{rawad.mezher@quandela.com}
\author{Shane Mansfield}
\affiliation{Quandela SAS, 7 Rue Léonard de Vinci, 91300 Massy, France}
\email{shane.mansfield@quandela.com}
%\homepage{http://quantum-journal.org}
%\orcid{0000-0003-0290-4698}
%\thanks{You can use the \texttt{\textbackslash{}email}, \texttt{\textbackslash{}homepage}, and \texttt{\textbackslash{}thanks} commands to add additional information for the preceding \texttt{\textbackslash{}author}. If applicable, this can also be used to indicate that a work has previously been published in conference proceedings.}
%\affiliation{Covestro Deutschland AG, Kaiser-Wilhelm-Allee 60, 51373 Leverkusen, Germany}
%\author{Marcus Huber}
%\affiliation{Institute for Quantum Optics \& Quantum Information (IQOQI), Austrian Academy of Sciences, Boltzmanngasse 3, Vienna A-1090, Austria}
%\orcid{0000-0003-1985-4623}
%\author{Christopher Granade}
%\affiliation{Microsoft Research, Quantum Architectures and Computation Group, Redmond, WA 98052, USA}
%\author{Johannes Jakob Meyer}
%\affiliation{Dahlem Center for Complex Quantum Systems, Freie Universität Berlin, 14195 Berlin, Germany}
%\orcid{0000-0003-1533-8015}
%\author{Victor V. Albert}
%\affiliation{Institute for Quantum Information and Matter \& Walter Burke Institute for Theoretical Physics, Caltech, Pasadena, CA 91125, USA}
%\orcid{0000-0002-0335-9508}
\maketitle

\begin{abstract}
%Noisy intermediate scale quantum devices are currently available across different  hardware platforms. Already, proof-of-principle experiments have shown that these devices can massively outperform their classical counterparts for specific tasks such as sampling. Furthermore, these devices are expected to be at the heart of many promising applications, notably through fields like variational quantum algorithms (VQAs) and quantum machine learning (QML). 
For near-term quantum devices, an important challenge is to develop efficient methods to certify that noise levels are low enough to allow potentially useful applications to be carried out.
% A recurring issue for current and near-term quantum devices is  developing efficient methods  to certify that  noise levels in these devices  are low enough to allow potentially useful applications to be carried out.
We present such a method tailored to photonic quantum devices consisting of single photon sources coupled to linear optical circuits coupled to photon detectors.
% Here, we present one such  method tailored to photonic quantum devices consisting of single photon sources, linear optical circuits, and single photon detectors.
It uses the output statistics of BosonSampling experiments with input size $n$ ($n$ input photons in the ideal case).
% Our method uses the output statistics of  BosonSampling experiments performed on these devices. 
%Each experiment is carried out by passing (ideally) $n$ single photons through $m=n^{2+\gamma}$ modes ($\gamma>0$ a constant) of a Haar random linear optical circuit; and statistics are collected by repeating these experiments many times, and for different configurations of Haar random linear optical circuits. 
We propose a series of benchmark tests targetting two main sources of noise, namely photon loss and distinguishability.
% We develop a series of tests to benchmark these statistics which target two main sources of noise, namely photon loss and non-identical single photons.
Our method results in a single-number metric, the Photonic Quality Factor, defined as the largest number of input photons for which the output statistics pass all tests.
% The output of our method is a single number metric, the photonic quality factor (PQF), defined as the largest number of input photons $n$ of  BosonSampling experiments which are  capable of passing all our developed tests.
We provide strong evidence that passing all tests implies that our experiments are not efficiently classically simulable,
by showing how several existing classical algorithms for efficiently simulating noisy BosonSampling fail the tests.
% We provide strong guarantees that passing all our tests implies our experiments are not efficiently classically simulable, by showing that many existing strategies for efficiently classically simulating BosonSampling fail our tests.
Finally we show that BosonSampling experiments with average photon loss rate per mode scaling as $o(1)$ and average fidelity of $ (1-o(\frac{1}{n^6}))^2$ between any two single photon states is sufficient to keep passing our tests.
% Finally, we show that BosonSampling experiments with average photon loss rate per mode scaling as $o(1)$ and average fidelity of $ 1-o(\frac{1}{n^6})$ between any two single photon states is sufficient to keep passing our tests.
Unsurprisingly, our results highlight that scaling in a manner that avoids efficient classical simulability will at some point necessarily require error correction and mitigation.
% Perhaps unsurprisingly, our results highlight the fact that notions of quantum error correction and/or mitigation must be introduced at some point when scaling up in order to improve performance suitably, and escape efficient classical simulability. 
\end{abstract}

%In the \texttt{twocolumn} layout and without the \texttt{titlepage} option a paragraph without a previous section title may directly follow the abstract.
%In \texttt{onecolumn} format or with a dedicated \texttt{titlepage}, this should be avoided.

%Note that clicking the title performs a search for that title on \href{http://quantum-journal.org}{quantum-journal.org}.
%In this way readers can easily verify whether a work using the \texttt{quantumarticle} class was actually published in Quantum.
%If you would like to use \texttt{quantumarticle} for manuscripts not yet accepted in Quantum, or not even intended for submission to Quantum, please use the \texttt{unpublished} option to switch off all Quantum related branding and the hyperlink in the title.
%By default, this class also performs various checks to make sure the manuscript will compile well on the arXiv.
%If you do not intend to submit your manuscript to Quantum or the arXiv, you can switch off these checks with the \texttt{noarxiv} option.
%On the contrary, by giving the \texttt{accepted=YYYY-MM-DD} option, with \texttt{YYYY-MM-DD} the acceptance date, the note ``Accepted in Quantum YYYY-MM-DD, click title to verify'' can be added to the bottom of each page to clearly mark works that have been accepted in Quantum.

\section{Introduction}
\label{intro}
Noisy intermediate-scale quantum (NISQ) \cite{NISQ} devices have now become available across a variety of different hardware platforms.
Proof-of-principle experiments have already demonstrated that such devices can massively outperform their classical counterparts for specific tasks such as sampling \cite{Arute2019,Zhong2020,Wu2021,Zhong2021}.
They also offer great promise for a range of near-term applications,
notably through variational quantum algorithms (VQAs) \cite{Tilly2021} and quantum machine learning (QML) \cite{QML}.

Our focus here will be on \emph{photonic} NISQ  devices
composed of three main components
% We envision these devices as being composed mainly of three components
\cite{KLM2001}:
$(i)$ $n$ single-photon sources (e.g.\ \cite{SGS2016,Nat2021,Singlephotrev}),
$(ii)$ a linear optical circuit of $m$ spatial modes composed of layers of configurable components like phase shifters and beam splitters (e.g.\ \cite{Clements2016,Reck94}),
and $(iii)$ single-photon detectors (e.g.\ \cite{Hadfield2009}).
The term NISQ is added to highlight that these devices are noisy with no error correction capabilities \cite{QEC}.
% We will let $n$ be the number of single-photon sources are input into $m$ spatial modes,
% and composed of a limited number $n$ of single photons and size $m$ of linear optical circuits.
% Note however that the certification techniques  we develop in this paper apply to \emph{any} size $m$ and $n$, including those of photonic NISQ devices.

Aside from BosonSampling experiments \cite{AA11,Lund14,HK17,MHP21} other more useful applications related to quantum machine learning can be implemented in the near-term with such devices, with the potential to demonstrate quantum-over-classical advantages \cite{UAM21,Gan21,ORCA,BSlearning}.  
It is therefore important to certify the correct functioning of these devices, and ensure the noise levels are low enough so as to maintain any potential quantum advantage offered. This is because sufficiently noisy quantum devices are known to be efficiently classically simulable in many cases, and thus offer no substantial quantum advantage  \cite{RSG18,OB18,ONF21,BMSQuantum,RaulGeom2020,Lim21}.

Although many techniques have been developed to certify quantum devices (see for example \cite{EH20,KR21} for a review), most are tailored to the gate-based model of quantum computing, and are therefore not very natural candidates for certifying our photonic NISQ devices. Indeed, the works of \cite{UAM21,Gan21,ORCA,BSlearning} show that it is possible to perform some machine learning tasks by means of a boson sampler and adaptive measurements, without the need to encode the photons as qubits nor to perform post-selections to implement non-deterministic two-qubit gates \cite{KLM2001}. Certifying photonic NISQ devices such as those used in \cite{UAM21,Gan21,ORCA,BSlearning} is what we will be concerned with in this work.

Our idea is similar in spirit to the Quantum Volume (QV) benchmark \cite{CBS+19,BK+19}. Our method consists of running several BosonSampling \cite{AA11} experiments with
varying input size, mode size, and parameter configurations.
% varying input and mode size, and varying configurations of linear optical circuits.
We then perform a set of tests on the output statistics of these experiments, and from the results of these tests compute our \emph{Photonic Quality Factor} (PQF), a single number metric which characterizes the \emph{average} performance of photonic NISQ devices. The PQF is the largest value of $n$ for which a photonic NISQ device can pass \emph{all} our developed tests. Intuitively, PQF can be understood as the largest number of input photons of a photonic NISQ device for which  \emph{reliable} BosonSampling experiments can be carried out using this device. 
BosonSampling can be viewed as the natural photonic analogue  of random quantum circuit sampling (RQCS) used to evaluate QV \cite{CBS+19}. However, the tests used in computing QV cannot  straightforwardly be imported to our setting to compute PQF. The main reason is that, although the statements of hardness of classical simulability of BosonSampling and RQCS are based on similar arguments \cite{MH17}, the so-called heavy output generation (HOG) test, which is central to computing QV \cite{CBS+19}, and the arguments underlying its validity as a certification tool for RQCS do not apply to BosonSampling \cite{AC16}. This motivates the need to develop other tests for photonic benchmarking.
% to compute our PQF benchmark, which we do in this work.
%In our method, each BosonSampling experiment is carried out by passing (ideally, in the absence of noise) $n$ single photons through $m=n^{2+\gamma}$ modes ($\gamma>0$ a constant) of a Haar random linear optical circuit, which (ideally) applies a Haar random $m \times m$ unitary transformation $U$ on the $n$ single photon input state, and measuring all the output modes. Statistics are collected by repeating these experiments many times, and for different configurations of Haar random linear optical circuits. 

The tests we propose are designed to target two main sources of noise in photonic NISQ devices, namely photon loss \cite{OB18} and distinguishability \cite{RMC18}. These noise sources when present either individually or simultaneously in sufficient quantity can destroy any potential quantum advantage \cite{RSG18,OB18,RMC18}. In total, we develop five tests, one for photon loss inspired from the results of \cite{RSG18}, and four for distinguishability. Three of the distinguishability tests are designed to detect variations in low order correlations between photons that can provide distintuishability witnesses, and are inspired from results in \cite{W16}.
The remaining distinguishability test is for variations in high order correlations, inspired by the results of \cite{S16}.
% We will define precisely each of our developed tests, the PQF, as well as the models of photon loss and distinguishability we adopt in the coming sections.
Note that the number of experiments needed for our method scales efficiently with the size of the boson sampler.

Passing all tests provides strong evidence that a photonic NISQ device can be used to perform applications showing a quantum-over-classical advantage. Indeed, as we will show in later sections, to the best of our knowledge all existing efficient classical simulation strategies \cite{RSG18,OB18,ONF21,RMC18} for noisy BosonSampling fail some or all of our tests after some (fixed) size of the boson sampler. Furthermore, we show that many efficient \emph{adversarial} classical algorithms for \emph{spoofing} BosonSampling such as those based on mean-field strategies \cite{Tichy14}, brute force permanent approximations \cite{UAM21,Gurvits2005,AH12}, or methods similar to those in \cite{Google21} will also fail some of our tests after a fixed size of the boson sampler. We stress however that our results \emph{do not} rule out the existence of efficient classical algorithms other than those studied here, which could produce statistics capable of passing the tests.

This paper is organized as follows. In section \ref{prelim} we will introduce various preliminary concepts and define our noise model. In section \ref{tests} we set out the tests to be performed to evaluate the PQF and discuss their sample complexity. In section \ref{PQF} we introduce PQF, and argue that it is a good metric for assessing the quality of a photonic NISQ device. In section \ref{experiment}, we design an experiment with PQF of $+\infty$; i.e.\ which would pass all tests indefinitely. Finally, we discuss our results in section \ref{discussion}.

\section{Preliminary concepts}
\label{prelim}
\subsection{Notation}
In this section, we briefly fix some notation to be used throughout the paper. For  $s \in \mathbb{R}$, $\left \lceil{s}\right \rceil \in \mathbb{Z}$  will denote the ceiling function applied to $s$; i.e.\ the smallest integer greater than or equal to $s$. 

We will use asymptotic notation to define various quantities such as precision of a test, or sample complexity of an experiment. Let $f,g: \mathbb{N} \to \mathbb{R}^{+}$.
% be two functions of some input size $n \in \mathbb{N}$.
We say that $f(n)=O(g(n))$ if there is a constant $c>0$ and a positive integer $n_0$, such that for all $n\geq n_0$, $f(n) \leq {c.g(n)}$. If $f(n)=o(g(n))$ , this means that $lim_{n \to \infty}\frac{f(n)}{g(n)}=0$. 
%We will often misuse notation and write $\frac{1}{O(f(n))}$ or $\frac{1}{o(f(n))}$ as $O(\frac{1}{f(n)})$ and $o(\frac{1}{f(n)})$. 

The expectation value of an observable $\mathbf{O}$ over some quantum state $|\psi\rangle$ will be denoted as $\langle \mathbf{O} \rangle_{\psi}:=\langle \psi |\mathbf{O}| \psi\rangle$, for simplicity we will omit the $\psi$ subscript and write the expectation value as $\langle \mathbf{O} \rangle$. For unitary matrices $U$, and $f$ a function on unitaries, the expectation value of $f$ over the Haar measure of the unitary group will be denoted as $\mathbf{E}_U(f(U))$.

For a random variable $X$ distributed according to some probability distribution, $\mathbf{E}(X)$ will denote the expectation value of $X$, $\mathsf{Var}(X)$ its variance, and $\sigma(X)=\sqrt{\mathsf{Var}(X)}$ its standard deviation.

For an $m \times m$ matrix $U$ and $m \in \mathbb{N}^{*}$, $\mathsf{Perm}(U)$ will denote the permanent of $U$ \cite{perm}.

The total variation distance (TVD) between two probability distributions
$D_1=\{p_x\}$ and $D_2=\{q_x\}$ will be denoted as $\|D_1-D_2\|$, and is given by
% $D_1:=\{p_i\}_{i=1,\dots,D}$ and $D_2:=\{q_i\}_{i=1,..D}$ where $p_i,q_i \geq 0$, $D$ is the size of the distributions, and $\sum_{i=1,..,D}p_i=\sum_{i=1,\dots,D}q_i=1$ will be denoted as $||D_1-D_2||$, and is given by

\begin{equation*}
   \|D_1-D_2\|=\frac{1}{2}\sum_{x}|p_x-q_x|. 
\end{equation*}

\subsection{BosonSampling}
BosonSampling, as originally proposed by Aaronson and Arkhipov \cite{AA11} is defined as a specific type of \emph{sampling} problem. We will begin by defining the problem setting as proposed in \cite{AA11}.
Let $m, n \in \mathbb{N}^{*}$ be positive integers with $m \geq n$. Let $\mathcal{S}_{m,n}$ be the set of all possible tuples $(s_1,\dots,s_m)$ of $m$ non-negative integers $s_i \in \mathbb{N}$, with $\sum_{i=1,..,m}s_i=n$. Let $U$ be an $m \times m$ Haar random unitary matrix. For a given fixed  $T=(t_1,\dots,t_m) \in \mathcal{S}_{m,n}$, and any given $S=(s_1,\dots,s_m) \in \mathcal{S}_{m,n}$, construct an $n \times n$ matrix $U_{T,S}$ by first constructing an $n \times m$ matrix $U_{S}$ from $U$ as follows: take $s_1$ copies of the first row of $U$, $s_2$ copies of the second row of $U$, and so on. Then, construct $U_{T,S}$ from $U_S$ by taking $t_1$ copies of the first column of $U_S$, $t_{2}$ copies of the second column of $U_{S}$,and so on. Let
\begin{equation}
 \label{eqbosonsampling1}
 P(S)=\frac{|\mathsf{Perm}(U_{T,S})|^2}{s_1 !\dots_m!}.
\end{equation}
It can be shown that $P(S) \in [0,1]$, and that the set $D_{U}=\{P(S) \mid S \in \mathcal{S}_{m,n}\}$ is a probability distribution over \emph{outcomes} $S$ \cite{AA11}. 

Let $\varepsilon \in [0,1]$ be a given \emph{precision}. Approximate BosonSampling can then be defined as the problem of sampling outcomes $S$ from a probability distribution $\tilde{D}$ such that
\begin{equation}
 \label{bosonsampling}
 \|\tilde{D}-D_{U}\| \leq \varepsilon.
\end{equation}
Exact BosonSampling has $\varepsilon=0$. For $m \gg n^2$, \cite{AA11} showed that no (efficient) $\mathsf{poly}(n,\frac{1}{\varepsilon})$-time classical algorithm can solve approximate BosonSampling, unless some complexity theoretic conjectures which are widely believed to be true turn out to be false. A similar result was shown for exact BosonSampling \cite{AA11}. On the other hand, BosonSampling can be solved efficiently on a photonic quantum device which is noiseless \cite{AA11}, or whose noise levels are low enough \cite{AB16,Arkhipov15,KK14}. We will refer to such a device as a \emph{boson sampler}.

Passing $n$ identical single photons through a lossless $m$-mode \emph{universal} linear optical circuit \cite{Reck94,Clements16} configured in such a way that it implements a Haar random unitary transformation $U$ and then measuring all the output modes of the circuit using perfect detectors samples outputs from $D_U$ \cite{AA11} (see Figure \ref{fig1}), where $T$ corresponds to the input configuration of single photons, and $S$ the output configuration.
% , with $t_i$ and $s_i$ for $i \in {1,\dots,m}$ being the number of photons occupying mode $i$ at the input and output ports of the boson sampler respectively.

\begin{figure}[htbp!]
\includegraphics[scale=0.5]{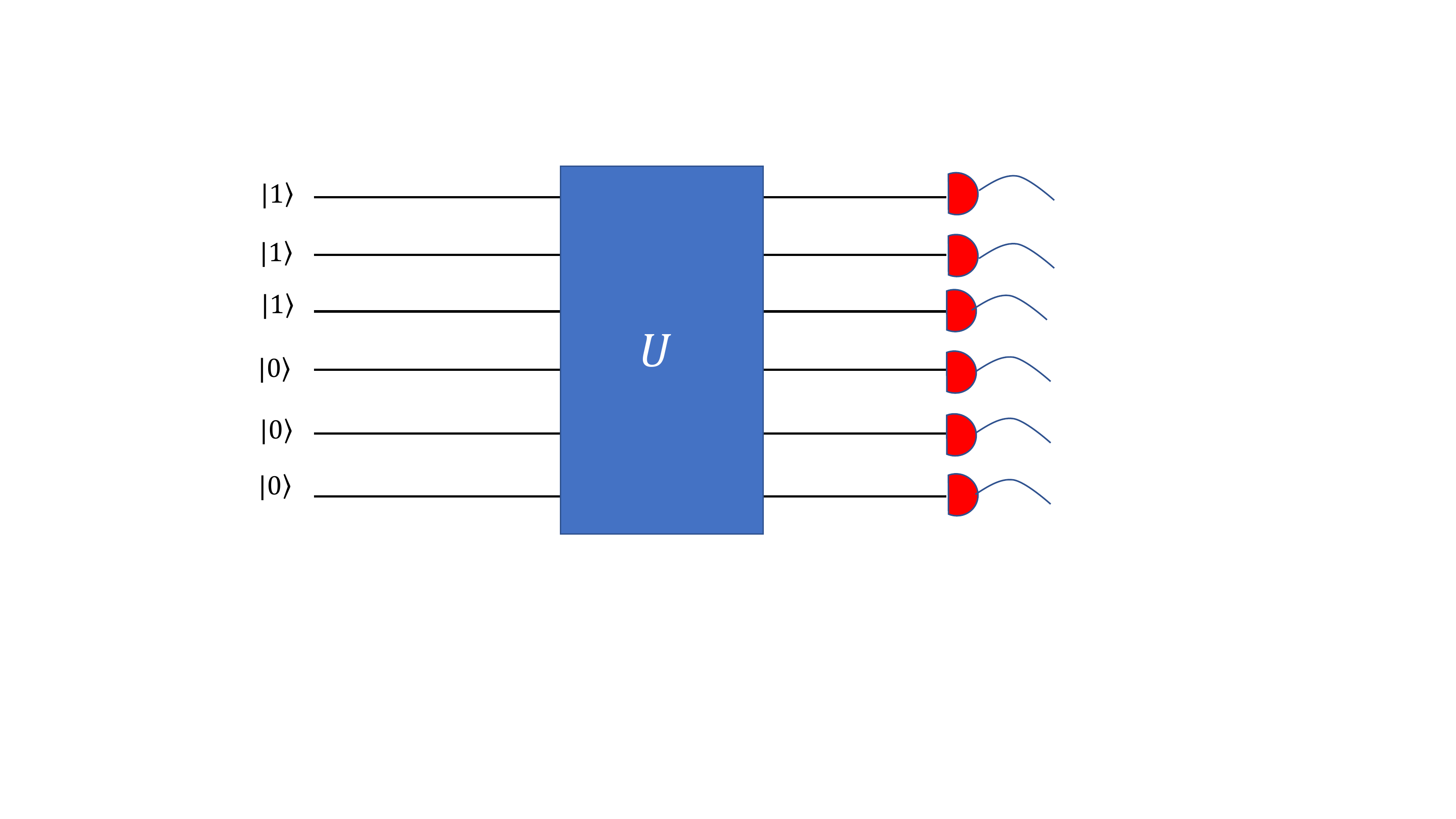}
%\centering
\caption{A depiction of a quantum device which implements BosonSampling. Here there are $n=3$ input photons and $m=6$ modes. The states $|1\rangle$ and $|0\rangle$ are Fock states; i.e.\ a single photon or the vacuum state, corresponding to no photon being present, respectively. %\cite{Kok07}.
$U$ is a Haar random $6 \times 6$ unitary matrix corresponding to the action of the linear optical circuit on the input.
% (red semi-ellipses with wire).
}
\label{fig1}
\end{figure}

Note that a universal linear optical circuit can be configured to implement \emph{any} $m \times m$ unitary chosen from the Haar measure, for example using the recipe of \cite{Haar17}. Also, note that for the case where $m \gg n^2$, it is known that the probability of \emph{collisions} -- events where two or more photons emerge at an output mode of a boson sampler -- can be neglected \cite{AA11}. This can be advantageous from an experimental point of view, as it would mean that the single-photon detectors are not required to be number-resolving. We will work with $m \gg n^2$ in this article, note however that the tests developed here can be adapted to work for any $m \geq n$.

 BosonSampling is one example of a family of sampling problems which can be used to show that quantum devices can massively outperform their classical counterparts \cite{MH17}. It is worth noting that other models of BosonSampling designed to overcome certain experimental difficulties have been proposed \cite{Lund14,MHP21,HK17}, although in this work we will focus on the proposal of \cite{AA11} as it fits naturally for the kinds of photonic NISQ devices we are concerned with.  Furthermore, two milestone \emph{Gaussian} BosonSampling experiments \cite{Zhong2020,Zhong2021} have very recently been performed which at present are believed to be classically intractable.
 
 Our aim here is to use the results of BosonSampling experiments run on photonic NISQ devices to assess the quality of these devices. This is similar in spirit to how \cite{CBS+19} use RQCS, another sampling problem strongly believed intractable for efficient classical devices \cite{MH17}, to assess the quality of their superconducting quantum devices. However, our techniques and tests are different to those used in \cite{CBS+19}. While useful quantum algorithms for photonic NISQ devices may well entail of other kinds of interferometers than boson samplers, the latter do model \emph{generic} instances of such algorithms, analogous to how random quantum circuits are generic instances of structured quantum circuits used to implement specific protocols \cite{CBS+19}. Therefore, similar to random circuits in the gate-based model \cite{CBS+19}, we would expect that statistics of BosonSampling experiments are good indicators of how more structured linear optical circuits will behave.
 
 %\textbf{-talk about ulysse ver. and why we chose not to use it (the value of $\varepsilon$ fixed for which BS  is hard to simulateclassically is not known, refer to oszmaniec and brod sec 5 or even aaronsons blog).}

 \subsection{Noise model}
 \label{secnoise}
 In this article we will be concerned with two sources of noise affecting photonic NISQ devices, namely photon loss and distinguishability of single photons, specifically pairwise distinguishability. We will not consider errors due to imperfect calibration of optical components such as beam splitters, as these can be studied and accounted for independently using techniques such as those in \cite{Pai19}. For now we also choose not consider other sources of error which may arise, e.g.\ due to dark counts in the detectors, cross-talk effects, or higher-order distinguishability, preferring to leave these for future work.
%  As previously, $m$ will denote the number of modes of the photonic NISQ device and $n$ the (ideal, if no loss) number of input photons of the device.
 
 Photon loss may occur at any point in the experimental setup, from source through fibres, interferometer to detector.
%  Photon loss occurs whenever a photon is lost in a photonic NISQ device. This loss can happen either at the level of the input, that is, while the photon passes through the various fibers on its way from the single-photon source to the input of the linear optical circuit, or inside the linear optical circuit itself, or at the output of the circuit (at the level of the single- photon detectors).
 We will make the following simplifying assumptions.
 \begin{enumerate}[label={\bfseries (A\arabic*)}]
 \item \label{A1} All modes suffer the same photon loss rate, and we will therefore model photon loss by a single parameter $\lambda \in [0,1]$, which is defined as the probability that a photon is lost in any given mode of the device.
 \item \label{A2} Our device can be modelled as a lossy input state, where each single photon can be lost independently with probability $\lambda$, followed by an ideal optical circuit, and ideal detectors.
 \end{enumerate}
Assumption \ref{A1} has been considered in many other works \cite{OB18,RSG18,ONF21}, and can be justified for the case of \emph{symmetric} linear optical circuits such as those in \cite{Clements16}, in which the number of 2-mode optical components (beam splitters and phase shifters) is roughly the same for each mode, each optical component has the same loss rate, and detectors also have identical loss rates.
% The intuition is that in such cases a photon sees on average the same number of lossy elements (beam splitters, phase shifters, detectors, \dots) in each mode.
Assumption \ref{A2} is based on the well-known fact that uniform losses in all modes commute with linear optical circuits, and therefore  losses happening at different parts of the circuit can be {commuted back} to the level of the input (see for example \cite{OB18} for a rigorous justification). 
 
 For $n$ single photon inputs photon distinguishability can be modelled in general by an $n \times n$ matrix $D:=(D_{ij})_{i,j \in \{1,\dots,n\}}$. For $|\phi_i\rangle$ and $|\phi_j\rangle$, the (pure) states of photons $i$ and $j$, we will define $D_{ij}:=\langle \phi_i|\phi_j\rangle$ \cite{RMC18}. More generally, one can characterize photon distinguishability for mixed photon states using convex combnations of elements of the form $\langle \phi_i|\phi_j\rangle$ \cite{S2020}.
%  If all the photons are identical, then $D$ is simply a matrix all of whose elements are $1$.
 For BosonSampling experiments, \cite{RMC18} showed that one can assume distinguishabilities to be positive reals without loss of generality. 
%  Assumption \ref{A3} was also used in \cite{RMC18,RSG18}. We will also, following \cite{RMC18,RSG18}, model our circuit as starting with an input of distinguishable photons with distinguishability $x$, with no additional distinguishability picked up at later parts of the circuit (assumption $(A4)$).
%  Therefore, the only assumption we make is that 
Following \cite{RMC18,RSG18} we will make the following simplifying assumptions about distinguishability.
 \begin{enumerate}[start=3,label={\bfseries (A\arabic*)}]
 \item \label{A3} All pairs of photons are {equally} distinguishable.
 \item \label{A4} We can model distinguishability as being constant throughout the circuit.
 \end{enumerate}
From \ref{A3}, we have that $D_{ij}=\langle \phi_i|\phi_j\rangle=x+(1-x)\delta_{ij}$ for all $i,j$, where $\delta_{ij}$ is the Kronecker delta function and $x \in [0,1]$ is a number we will refer to as the \emph{distinguishability}. We will refer to $F=x^2$ as the \emph{average fidelity} between any two single photon states. 
 
As a concluding remark, note that $x$ and $\lambda$ are  \emph{average} quantities, where the average is over all possible parameter configurations of linear optical circuits (i.e.\ over the Haar measure on $m \times m$ unitaries). We will also assume the following.
  \begin{enumerate}[start=5,label={\bfseries (A\arabic*)}]
 \item \label{A5} There is small variance in loss and distinguishability over varying parameter configurations.
 \end{enumerate}
This assumption is reasonable from an experimental point of view. For example, photon loss is very weakly dependent on the specific configuration (angles of beam splitters and phase shifters) of the beam splitters and phase shifters in a linear optical circuit, and depends mainly on the \emph{depth} of the circuit.
 %for example it is known that, for photonic NISQ devices, most losses occur either at the level of the single photon source or at the level of the detectors, rather than at the level of the circuit itself (which is not very deep in the case of photonic NISQ devices). 
 Assumption \ref{A5} is one of the main reasons we believe the proposed tests can be used to infer the performance of linear optical circuits with fixed configurations performing specific computations.
 
 \section{The loss and distinguishability tests}
 \label{tests}
 In what follows, we will be working in the \emph{no-collision} regime of BosonSampling ($m \gg n^2$), where the probability of more than one photon being detected at any output mode of the boson sampler is negligible \cite{AA11}. To this end, all experiments are based on performing BosonSampling  with $n$ (lossy and distinguishable) single photon inputs and an $m=n^{2+\gamma}$ mode universal linear optical circuit \cite{Reck94,Clements16}, where $\gamma>0$ is a constant. The input of the boson sampler (in the ideal case where no losses happen and photons are identical) is the Fock state $|1_1,1_2,\dots,1_n,0_{n+1},\dots,0_{m}\rangle$ where the $n$ single photons are placed in the first $n$ modes, one per each mode. The protocol for collecting the experimental data needed for the tests is as follows. 
 
 \begin{enumerate}
     \item  Configure the universal linear optical circuit to implement an $m \times m$ Haar random unitary $U$, using the recipe of \cite{Haar17} for example.
      \item  Run $K^{'}$ BosonSampling experiments, and collect the output statistics. Where each BosonSampling experiment consists in passing $n$ single photons through $U$, and measuring all output modes of $U$ using single-photon detectors. 
      \item  Repeat steps 1) and 2) for $K^{''}$ different values of $U$.
 \end{enumerate}
 We will discuss  how large $K^{'}$
 and $K^{''}$ need to be, in function of $n$, in later parts of this section. We will postpone talking about  the significance of each of the developed tests until section \ref{PQF}, and will only proceed in this section  by describing in detail how to perform each test.
 \subsection{ A test for loss}
 \label{sectloss}
 
 The first test we will describe here is for photon loss. We will call $\lambda_{U}$ the photon loss for BosonSampling experiments involving a fixed choice of $U$. From our model for photon loss (see section \ref{secnoise}), the probability of $l$ photons being lost is given by
 \begin{equation}
     \label{eqpl}
     P(l)= {n \choose n-l}(1-\lambda_U)^{n-l}\lambda_U^l,
 \end{equation}
 with $l \in \{0,\dots,n\}$.
These probabilities are those of the well known binomial distribution with the mean number of lost photons being $\mathbf{E}(l)=n\lambda_U$ and $\mathsf{Var}(l)=n\lambda_U(1-\lambda_U)$. We start by computing $\lambda_U$ for each of the $K^{''}$ values of $U$ individually. To do this, we look at the experimental statistics of the $K^{'}$ BosonSampling experiments involving a fixed $U$, and from these experiments we compute for $l=0,\dots,n$ estimates of $P(l)$ as
\begin{equation}
    \tilde{P}(l)=\frac{n_{l,U}}{K^{'}},
\end{equation}
 where $n_{l,U}$ is the number of experiments where we detected $l$ lost photons, out of $K^{'}$ total experiments for a given fixed $U$. For large enough value of $K^{'}$, $\tilde{P}(l) \approx P(l)$. We then use these values of $\tilde{P}(l)$ to compute $\tilde{\lambda}_U$, the estimate of $\lambda_U$ as follows
 \begin{equation}
     \tilde{\lambda}_U=\frac{\sum_{l=0}^{n}l\tilde{P}(l)}{n} \approx \frac{\mathbf{E}(l)}{n}=\lambda_U.
 \end{equation}
 We repeat this procedure for each of the $K^{''}$ different values of $U$, then compute 
 \begin{equation}
 \tilde{\lambda}=\frac{\sum_{U}\lambda_U}{K^{''}}.
 \end{equation}
 Again, for large enough $K^{''}$,
 \begin{equation}
 \tilde{\lambda} \approx \mathbf{E}_U(\lambda_U):=\lambda.
 \end{equation}
 Now, we are in a position to describe our loss test, which we will henceforth call $t_{loss}$. It is simply the following
 \begin{itemize}
 \item ($t_{loss}$): For varying values of $n$, check whether $\lambda \leq o(1)$.
 \end{itemize}
 
 A key point to note is that $\lambda$ \emph{need not be} a decreasing function of $n$, as long as it is \emph{below} a decreasing function of $n$ (i.e.\ $o(1)$), then we can keep on passing $t_{loss}$. This remark is important in the context of photonic NISQ devices since, because of the absence of any error correction, one would expect $\lambda$ to be non-decreasing with $n$. The technological challenge is therefore finding a way of managing this non-decrease as we scale up, to keep on passing $t_{loss}$ up to a value of $n$ where a potentially advantageous quantum computation could be carried out on the NISQ device. The same reasoning holds for all the other tests developed in the coming subsections.
 %Where $f(n)=o(1)$ is some decreasing function of $n$. 
 %That is, check whether $\lambda$ is less than or equal to a decreasing function of $n$ which approaches zero as $n$  becomes larger.
 \subsection{Tests for distinguishability}
 
 \subsubsection{Low order correlation tests}
 Our low order correlation tests are inspired from the results in \cite{W16}, showing that the first, second, and third statistical moments of a set of 2-mode correlation functions computed in BosonSampling experiments are enough to distinguish the behaviour of $n$ identical photons from the behaviour of $n$ \emph{simulated} Bosons \cite{Tichy14} (particles in an efficiently classically simulable version of BosonSampling, based on \emph{mean field} approaches, designed to mimic certain bunching effects \cite{RudolphBS} found in ideal BosonSampling with identical particles), as well as from that of $n$ totally distinguishable Bosons ($x=0$). BosonSampling experiments carried out with simulated Bosons or totally distinguishable particles are efficiently simulable classically \cite{Tichy14,RMC18}. The  2-mode correlation function we will study here is the following \cite{W16}
 \begin{equation}
     \label{eqCdataset1}
     C_{ij}=\langle \textbf{n}_i\textbf{n}_j \rangle -\langle \textbf{n}_i \rangle \langle \textbf{n}_j \rangle.
 \end{equation}
 Where $\textbf{n}_i$ (resp. $\textbf{n}_j$) is the \emph{number operator} of mode $i$ (resp. $j$), which quantifies how many photons are in mode $i$ (resp. $j$) of the output of the boson sampler (see \cite{Kok07,W16} for more details). Following \cite{W16}, we refer to the set of all such correlators 
 \begin{equation}
     \label{eqCdataset2}
      C=\{C_{ij}\}_{i,j=1,\dots,m},
 \end{equation}
 as the $C$-data set. Note that the number of elements of $C$ is $|C|={m \choose 2}$. For a fixed $U$, and a fixed $l \in \{0,\dots,n-2\}$, we will compute estimates  $\tilde{C}_{ij,l,U}$ of $C_{ij,l,U}$  for all $i,j$ using the statistics of $n_{l,U}$ experiments. The indices $l$ and $U$ in $C_{ij,l,U}$ are to indicate at which number of lost photons calculations are carried out, as well as for which choice of $U$. Then, we compute the following  for all $i,j \in \{1,\dots,m\}$ such that $i<j$
 \begin{equation}
     \label{eqcijhaar}
     \frac{\sum_{U}\tilde{C}_{ij,l,U}}{K^{''}} \approx \mathbf{E}_U(C_{ij,l}),
 \end{equation}
 where the approximation in Equation (\ref{eqcijhaar}) holds for large enough $n_{l,U}$, $K^{'}$ and $K^{''}$. We will then use the values of the \emph{averaged} $C$-data set
 \begin{equation}
     \label{eqcdatasethaar}
     C_{H}:=\{\mathbf{E}_U(C_{ij,l})\}_{i,j=1,\dots,m},
 \end{equation}
 to compute the following quantities directly related to the first, second, and third order statistical moments of $C_H$ \cite{W16}
 \begin{equation}
     \label{eqNM}
     NM=\frac{m^2}{n}\mathbf{E}(\mathbf{E}_U(C_{ij,l})):= \frac{m^2}{n}\frac{\sum_{i,j}\mathbf{E}_U(C_{ij,l})}{{m \choose 2}},
 \end{equation}
 \begin{equation}
     \label{eqCV}
     CV=\frac{\sqrt{\mathbf{E}(\mathbf{E}_U(C^2_{ij,l}))-(\mathbf{E}(\mathbf{E}_U(C_{ij,l})))^2}}{\mathbf{E}(\mathbf{E}_U(C_{ij,l}))},
 \end{equation}
 and
 \begin{equation}
     \label{eqskewness}
     S=\frac{\mathbf{E}(\mathbf{E}_U(C^3_{ij,l}))-3\mathbf{E}(\mathbf{E}_U(C^2_{ij,l}))\mathbf{E}(\mathbf{E}_U(C_{ij,l}))+2\mathbf{E}(\mathbf{E}_U(C_{ij,l})))^3}{\big (\sqrt{\mathbf{E}(\mathbf{E}_U(C^2_{ij,l}))-(\mathbf{E}(\mathbf{E}_U(C_{ij,l})))^2}\big )^{3}}.
 \end{equation}
 $NM$ stands for normalized mean, $CV$ for coefficient of variation and $S$ for skewness. Note that Walschaers et al. \cite{W16} compute the analytical values of these quantities for identical photons ($x=1$) by using techniques from random matrix theory to formally average over the Haar measure. We will denote these theoretical quantities as $NM_{Th,id}$, $S_{Th,id}$, and $CV_{Th,id}$.
 
 Our low order correlation tests consist of verifying the following for $l \in \{0,\dots,n-2\}$
 \begin{itemize}
     \item $(t_{d_1}$): For varying values of $n$, check whether $|NM-NM_{Th,id}| \leq o(\frac{1}{n})$.
 \item $(t_{d_2}$): For varying values of $n$, check whether $|CV-CV_{Th,id}| \leq o(\frac{1}{n})$.
 \item $(t_{d_3}$): For varying values of $n$, check whether $|S-S_{Th,id}| \leq o(\frac{1}{n})$.
 \end{itemize}
 
 Implicit in our description of the tests are the two following remarks. The first being that BosonSampling experiments with $l$ lost photons can be thought of as a BosonSampling with $n-l$ input photons, where these photons are randomly permuted in ${n \choose n-l}$ ways in the first $n$ modes, each such permutation is equiprobable to appear per each run of a BosonSampling experiment \cite{AB16}. The second remark is that these permutations are irrelevant from the point of view of averaging over the Haar measure. That is, each  permutation by itself will give rise to the \emph{same} averaged $C$-data set $C_H$ when used as a (fixed) input for BosonSampling experiments averaged over Haar random $U$'s. Therefore, as is the case for our experiments, a statistical mixture of these permutations (each appearing equiprobably) will also give the same $C_H$. The reason behind the fact that each permutation gives the same $C_H$ is that each permutation is related to the \emph{canonical} permutation $|1_1,\dots,1_{n-l},0_{n-l+1},\dots,0_{m}\rangle$ by a unitary $U_{route}$, and the Haar measure is \emph{invariant} under multiplication by a fixed unitary $U_{route}$. In Appendix \ref{approuting} we provide a constructive procedure to implement $U_{route}$ for any permutation using linear optical circuits. The same arguments hold for statistics collected for the high order correlation test in the next section.

 \subsubsection{A high order correlation test} 
 Our high order correlation test is inspired from the work of \cite{S16}. Let $1 \leq K \leq m$. The  basic idea behind this test is to compute, for a loss $l \in \{0,\dots,n-1\}$ and a fixed $U$, an estimate (  $\tilde{P}_{K,l,U}(0_{K+1}\dots0_{m})$) of the  probability of finding all $n-l$ particles in the first $K$ output modes of the boson sampler
 \begin{equation}
 \label{eqph}
 P_{K,l,U}(0_{K+1}\dots0_{m})=\sum_{s_1,\dots,s_K}P((s_1,\dots,s_K,0_{K+1},\dots0_{m})),
 \end{equation}
 where $s_1+\dots+s_K=n-l$. The estimate $\tilde{P}_{K,l,U}(0_{K+1}\dots0_{m})$ is computed using the statistics of the $n_{l,U}$ BosonSampling experiments for fixed $l$ and $U$, and converges to the value in Equation (\ref{eqph}) for large enough $n_{l,U}$. Then, we compute 
 \begin{equation}
 \label{eqphhaar}
 \sum_{U}\frac{\tilde{P}_{K,l,U}(0_{K+1}\dots0_{m})}{K^{''}} \approx \mathbf{E}_{U}(P_{K,l}(0_{K+1}\dots0_{m})).
 \end{equation}
Choose $m-K=n-1$ so that $(n-l)(m-K)<n(m-K) <<m$. For this choice of $K$, and for the case of identical photons, \cite{S16} shows that
\begin{equation}
    \label{eqtheoravph}
    \mathbf{E}_{U}(P_{K,l}(0_{K+1}\dots0_{m}))_{Th,id}=1-O(\frac{(n-l)(m-K)}{m}).
\end{equation}
\cite{S16} also shows that, for the case of photons with distinguishability $x$ (where $x \approx 1$), we have (see appendix \ref{appC} )
\begin{multline}
    \label{eqtheoravph2}
    \mathbf{E}_{U}(P_{K,l}(0_{K+1}\dots0_{m}))_{Th,id}-\mathbf{E}_{U}(P_{K,l}(0_{K+1}\dots0_{m}))_{Th,x}=\\(1-F)\frac{(n-l-1)(n-l)}{m}\mathbf{E}_{U}(P_{K,l+1}(0_{K+1}\dots0_{m}))_{Th,id}=\\ (1-x^2)\frac{(n-l-1)(n-l)}{m}\mathbf{E}_{U}(P_{K,l+1}(0_{K+1}\dots0_{m}))_{Th,id}.
    \end{multline}
    The indexes $id$ and $x$ in the above equations are used to differentiate the case of identical photons from that of photons with distinguishability $x$. As before, the $Th$ index indicates an analytical expression obtained from formally averaging over the Haar measure.
    
    Our test for high order correlations can now be described as follows for $l\in \{0,\dots,n-1\}$
    \begin{itemize}
        \item $(t_{d_4})$: For varying values of $n$, check whether \\ $|\mathbf{E}_{U}(P_{K,l}(0_{K+1}\dots0_{m}))-\mathbf{E}_{U}(P_{K,l}(0_{K+1}\dots0_{m}))_{Th,id}| \leq o(\frac{1}{n^{\gamma}})$.
    \end{itemize}
    
    As a concluding remark, note that the probabilities used in evaluating the test $(t_{d_4})$ \cite{S16} are a special case of a more general set of such probabilities developed, by the same author of \cite{S16}, in \cite{S2020}. The main differences being that calculations in \cite{S2020} include dark counts of the detectors, whereas we do not consider this here, and also the quantities computed in \cite{S2020} are for any $m \geq n$ and include the effect of loss as a binomial distribution which enters into computing these quantities. In our case however, we work only in the $m \gg n^2$ case, and compute probabilities for each fixed loss $l$, by viewing this loss as a lossless BosonSampling with $n-l$ input photons.
    
    \subsection{Sample complexity of the tests}
    \label{secsamplcomp}
    In this section, we discuss the sample complexities of our developed tests, i.e.\ how many BosonSampling experiments need to be done to estimate to a desired accuracy our quantities of interest. We will show that the number of these experiments scales efficiently with $n$. Rather than go through each quantity individually, we present an argument which holds in general for all the quantities needed in our tests, and for averaging both over the Haar measure over all unitaries, or over the number of experiments for a fixed unitary. The starting point of our argument is Chebyshev's inequality \cite{Chebyshev}. Note that in order to use this inequality we will assume that the random variables underlying our quantities of interest are independent. Practically, this means the outcomes of two different BosonSampling experiments are independent, which is a reasonable assumption. Let $X_{i}$ for $i=1,\dots,L$ be independent random variables chosen from a probability distribution with mean $\mathbf{E}(X)=\mu$ and variance $\mathsf{Var}(X)=\sigma^2$. Chebyshev's inequality gives the following bound on the estimate $\frac{\sum_{i}X_i}{L}$.
    \begin{equation}
        \label{eqchebyshev}
        Pr(|\frac{\sum_{i}X_i}{L}-\mu| < \epsilon) \geq 1-\frac{\sigma^2}{L\epsilon^2},
    \end{equation}
    for all $\epsilon \in \mathbb{R}^{+*}$.

    For example, if we are computing $\tilde{\langle \mathbf{n}_u\mathbf{n}_j \rangle }_{uj,l,U}$, an estimate of  $\langle \mathbf{n}_u\mathbf{n}_j \rangle_{uj,l,U}$ for some fixed loss $l$, and fixed $U$, then $L=n_{l,U}$, $\mu=\langle \mathbf{n}_u\mathbf{n}_j \rangle_{uj,l,U}$, $X_i$ is the sum of the occupancy of modes $u$ and $j$ at experiment $i$,  and $i \in \{1,\dots,n_{l,U}\}$. Chebyshev's inequality in this case reads 
    \begin{equation}
         Pr(|\frac{\sum_{i}X_i}{n_{l,U}}-\langle \mathbf{n}_u\mathbf{n}_j \rangle_{uj,l,U}| < \epsilon) \geq 1-\frac{\sigma^2}{n_{l,U}\epsilon^2}.
    \end{equation}
    Similarly, we can form Chebyshev inequalities for all our quantities of interest, including for averaging over the Haar measure.
    
    Looking at Equation (\ref{eqchebyshev}), the trick is to notice that, for a variance $\sigma^2$ which is \emph{bounded} (i.e.\ $\sigma^2 \leq O(1)$), one can in principle compute the estimate to arbitrary precision $\epsilon$ with arbitrary confidence $1-\frac{\sigma^2}{L\epsilon^2}$, by an interplay between the choice of $L$ and $\epsilon$. If we would like an inverse polynomial in $n$ precision $1/poly(n)$, with $1-1/poly(n)$ confidence, then we can choose $\epsilon=1/poly(n)$, and $L\epsilon^2 \geq O(1)$ or equivalently, $L \geq O(1).poly^2(n)=poly(n)$. Thus, for bounded variance, the number of experiments needed to reach a good precision scales efficiently with the system size. All we have to do now is to show that the variance of all our quantities of interest is bounded. For probabilities and photon loss parameters, this is trivially true, since these quantities lie in the interval $[0,1]$. What remains is to show that the distributions over the quantities $C_{ij}$ have bounded variance. Since we are in the no-collision regime, then $n_{i},n_{j} \in \{0,1\}$. This immiediately implies,using Equation (\ref{eqCdataset1}) and a triangle inequality, that $|C_{ij}| \leq 2$. Since we have shown the $C$-data set is composed of bounded quantities, then it immediately implies that the variance of the $C$-data set is bounded. This can be seen directly by noting that $\mathsf{Var}(X)=\mathbf{E}(X^2)-\mathbf{E}(X)^2$ and then using the monotonicity of the expectation value ( if $X<a$, then $\mathbf{E}(X) <a$ when it exists). Note that, as seen previously, the $C$-data set is composed of ${m \choose 2}=poly(n)$ elements $C_{ij}$, and since we have just shown that each $C_{ij}$ can be computed to $1/poly(n)$ precision using $poly(n)$ samples, then computing the entire $C$-data set also needs $poly(n)poly(n)=poly(n)$ samples.
    
    %To summarize, we have shown that in order to get our quantities of interest up to $1/poly(n)$ precision of their actual value, we need an efficient $poly(n)$ number of BosonSampling experiments. Such a precision is sufficient for the purpose of our tests, as for example we would only need $o(1/n^{\gamma})$ precision to evaluate $t_{d_4}$. 
    So far, we have discussed estimating quantities of interest in the case where we have a \emph{fixed} number of lost photons $l$, and showed that this can be done efficiently by using $poly(n)$ samples. However, since the losses follow a binomial distribution, some events, such as $l=0$ happen with \emph{exponentially} low probability $(1-\lambda)^n$, when $\lambda n\gg 1$, which is the case we are studying ($\lambda= o(1) \gg \frac{1}{n}$). Thus, collecting enough ($n_{0,U}=poly(n)$) statistics for this instance might take an exponential number of runs $K^{'}$ (see beginning of section \ref{tests}). Fortunately, by virtue of assumptions \ref{A2} and \ref{A4} (see section \ref{secnoise}) stating that all modes are equally lossy and all photons equally distinguishable, we could perform our tests only for  values of $l$ in the vicinity of the mean number of lost photons $\lceil \lambda n \rceil$ instead of for all values of $l$. Indeed, for the binomial distribution, most of the values of $l$ appearing in the $K^{'}$ experiments will belong to the interval $L_0=[\lceil \lambda n \rceil - l_0, \lceil \lambda n \rceil+l_0]$, where $l_0=o(\lceil \lambda n \rceil)$, which can be computed for example using Chernoff bounds \cite{Chebyshev}, depending on the level of confidence. If we focus on losses in the interval $L_0$, a $poly(n)$ number of experiments $K^{'}$ would be sufficient to get $n_{l,U}=poly(n)$ instances for each $l \in L_0$, allowing us to approximate to $1/poly(n)$ precision all of our desired quantities. Thus, the number of experiments $K^{'}$ needed is efficient. Note that, in order to compute $\lambda_U$ for $l \in L_0$, we can for example extrapolate it from the curve $P(l)={n \choose n-l}(1-\lambda_U)^{n-l}\lambda_U^{l}$ drawn for $l \in L_0$, instead of computing $\mathbf{E}(l)$ as is done in section \ref{sectloss}.

    \bigskip
    \bigskip
    To conclude this section, we remark that other techniques of certification of BosonSampling exist, such as those based on computing the $\mathsf{TVD}$ between the experimental BosonSampling distribution and the ideal one \cite{CGK+21}. Although arguably more straightforward to implement (in terms of needing only one test, that of evaluating the $\mathsf{TVD}$) and making no assumptions on the error model, we do not use such techniques mainly for two reasons. The first being that, unlike other proposals for computational quantum speedup such as those in \cite{BMS16}, for BosonSampling it is not known what the upper bound $\epsilon$ on the $\mathsf{TVD}$ between ideal and experimental  distributions should be in order to claim \emph{speedup} in the sense that no efficient classical algorithm can produce a distribution with $\mathsf{TVD} \leq \epsilon$, whereas for instantaneous quantum polynomial-time (IQP) circuits \cite{BMS16}, $\epsilon=1/192$ for example. The second, and arguably more severe, reason is that for noisy BosonSampling, the contributions of higher order interference terms (which are hard to compute) to the permanent seem to be suppressed exponentially \cite{Google21}. This means that one could in principle compute to good accuracy lower order interference terms (which are easy to compute), without caring too much about the accuracy of the higher order ones, and use these in order to \emph{spoof} noisy BosonSampling, by producing a $\mathsf{TVD}$ better than that of noisy BosonSampling. Indeed, this was what was done in \cite{Google21} for the case of Gaussian BosonSampling. In our case, the test $t_{d_4}$ can only be computed to the required precision by \emph{accurately} computing these high order interference terms (high order marginal probabilities), making spoofing attacks such as those in \cite{Google21} unlikely to pass $t_{d_4}$. More about this is to be said in coming sections.
    
    \section{The Photonic Quality Factor (PQF)}
    \label{PQF}
    \subsection{Definition}
    Having defined each of the tests $t_{loss}$, and $t_{d_1}$to $t_{d_4}$; in this section we introduce PQF, and argue that it is a good metric for assessing the performance of photonic NISQ devices, by showing how the tests for PQF are designed to expose many efficient classical simulation strategies for BosonSampling \cite{RSG18,OB18,ONF21,RMC18,Tichy14,UAM21,Gurvits2005,AH12,Google21}.
    
    To define PQF, we will first let $o(1)=O(\frac{1}{n^{\epsilon_1}})$ in the definition of $t_{loss}$, $o(\frac{1}{n})=O(\frac{1}{n^{1+\epsilon_2}})$ in the definitions of $t_{d_1}$ to $t_{d_3}$, and $o(\frac{1}{n^\gamma})=O(\frac{1}{n^{\gamma+\epsilon_3}})$ in the definition of $t_{d_4}$, where $\epsilon_i \in \mathbb{R}^{+*}$ for $i=1,2,3$ (see section \ref{tests} for definitions of $t_{loss}$, and $t_{d_1}$to $t_{d_4}$). Note that the larger the values of $\epsilon_i$, the more stringent the tests are, since an increase in these values means that we are testing for quantities which are closer to their ideal (lossless, identical photons) values. Let $\{\epsilon\}=\{\epsilon_1,\epsilon_2,\epsilon_3\}$. We will use the term noisy boson sampler of size $n$ to refer to a boson Sampler with $n$ input photons, and having some level of photon distinguishability and photon loss. 
    %We will mean by a noisy BosonSampling of size $n$ an experiment (or set of experiments) done with a noisy Boson sampler with $n$ input photons having  an average photon loss $\lambda$ and average distinguishability $x$ (see section \ref{secnoise}). 
    %Unless otherwise specified, the number of modes of a BosonSampling of size $n$ will always be $m=n^{2+\gamma}$.
    
    \begin{definition}
    \label{def1}
    $\emph{PQF}_{\gamma,\{\epsilon\}}$ is the maximum value of $n$ such that we can construct a noisy boson sampler of size $n$ with $m=n^{2+\gamma}$ modes whose collected output statistics pass all the tests $t_{loss}$ and $t_{d_1}$ to $t_{d_4}$.
    \end{definition}
     Note that Definition \ref{def1} captures a family of single number metrics parametrized by $\gamma$ (related to the mode size), and $\{\epsilon\}$ related to the tests $t_{loss}$ and $t_{d_1}$ to $t_{d_4}$ and how stringent we would like these tests to be. Also note that, by construction, $\textrm{PQF}_{\gamma,\{\epsilon\}}=+\infty$ for an ideal (noiseless) boson sampler with $\lambda=0$ and $x=1$. Finally, note that for a fixed $\gamma$ and $\{\epsilon\}$, one could use $\textrm{PQF}_{\gamma,\{\epsilon\}}$ as a metric to compare across different photonic hardware (as long as each of these hardware is based on single photon sources, linear optical circuits, and single photon detectors). It is an interesting open question to determine the parameters $\gamma$ and $\{\epsilon\}$, as well as the multiplicative constants (i.e.\  determining  $c$ in $O(f(n))=cf(n)$, for all tests), which are most suited for our hardware at Quandela, however this question is not pursued in this work.
     
     In the coming sections, we will show how several classical algorithms used to efficiently simulate noisy BosonSampling will fail one or more of our tests.  %discussion, we will consider $\gamma$ and $\{\epsilon\}$ to be fixed in the next subsections.
     \subsection{Efficient classical algorithms, and how they fail our tests}
     
     %%%%% TO AVOID CONFUSION, ALWAYS WORK WITH $l=\lambda.n$, THE MEAN LOSS, WHEN DISCUSSING HOW CLASSICAL ALGORITHMS FAIL THE TEST $(t_{d_4})$
     \subsubsection{Simulated bosons and completely distinguishable particles}
     BosonSampling performed with completely distinguishable particles ($x=0)$ or with simulated bosons \cite{Tichy14} is efficiently simulable classically. It is therefore important that these two types of particles produce statistics which cannot pass our tests (for all $n \geq n_0$, where $n_0$ is some fixed size of a boson sampler). For the case of Bosonsampling with $n$ particles, neglecting losses, with $m\gg n^2$, and for both types of particles as well as for identical photons, Walschaers et al. \cite{W16} compute the values of $NM$, $CV$, and $S$  formally averaged over all Haar unitaries as
     \begin{equation*}
         NM_{Th,id} \approx -o(\frac{1}{n})-1,
     \end{equation*}
     \begin{equation*}
         CV_{Th,id}\approx \frac{2}{n}-1,
     \end{equation*}
     \begin{equation*}
         S_{Th,id} \approx 2-\frac{30}{n},
     \end{equation*}
     \begin{equation*}
         NM_{Th,d} \approx o(\frac{1}{n^2})+1,
     \end{equation*}
     \begin{equation*}
         CV_{Th,d}\approx -\sqrt{\frac{3}{n}},
     \end{equation*}
     \begin{equation*}
         S_{Th,d}\approx -\frac{26}{\sqrt{27}}\sqrt{\frac{1}{n}},
     \end{equation*}
     \begin{equation*}
         NM_{Th,sb} \approx NM_{th,id},
     \end{equation*}
     \begin{equation*}
         CV_{Th,sb} \approx \frac{1}{2n}-1,
     \end{equation*}
     and 
     \begin{equation*}
         S_{Th,sb} \approx 2-\frac{21}{n}.
     \end{equation*}
     Where, as before, the subscript $Th$ is used to indicate that the value is calculated analytically, and subscripts $id$, $d$, and $sb$ are used to differentiate between quantities computed for identical photons ($id$), completely distinguishable particles ($d$), and simulated bosons ($sb$). 
     
     Let $l \neq n$ be a given number of lost photons. In this case, the expressions for $NM$, $CV$, and $S$ for each of the above mentioned particle types can be obtained by replacing $n$ with $n-l$ in the above equations. 
     
     Consider the case of completely distinguishable particles. Performing the $t_{d_1}$ test for this case by using the expressions given in the above equations, we get
     $$|NM_{Th,d}-NM_{Th,id}|=2+o(\frac{1}{n-l})\gg o(\frac{1}{n}).$$
     Thus, completely distinguishable particles fail our $t_{d_1}$ test. By similar arguments, one could  show that these particles also fail the tests $t_{d_2}$ and $t_{d_3}$.
     
     For simulated bosons, performing the $t_{d_2}$ test yields
     $$|CV_{Th,sb}-CV_{Th,id}|=\frac{3}{2(n-l)} \geq \frac{3}{2n}\gg o(\frac{1}{n}).$$
     Thus, simulated bosons fail the $t_{d_2}$ test. One can similarly show that they also fail the $t_{d_3}$ test.
     %Simulated Bosons were first considered in \cite{

     \subsubsection{ Efficient classical algorithms for lossy BosonSampling}
     We will examine two classical simulation algorithms for lossy BosonSampling. The $poly(n)$-time algorithm of \cite{OB18} which \emph{weakly} simulates (samples from) lossy BosonSampling efficiently up to a precision $\epsilon(n)$ in the $\mathsf{TVD}$, where $\epsilon(n)$ is dependent on the size of the boson sampler. The second algorithm is that of 
     \cite{ONF21}, which can weakly simulate lossy BosonSampling up to \emph{any} $\mathsf{TVD}$ error $\epsilon$ by means of a $poly(n,\frac{1}{\epsilon})$-time algorithm, by using techniques based on matrix product states. Both of these efficient classical algorithms work in a regime of losses given by $$\lambda \approx 1-o(\frac{1}{\sqrt{n}}).$$
     Evidently, $$1-o(\frac{1}{\sqrt{n}})\gg o(1),$$ after some value of $n$, which means that the regime in which these classical algorithms operate fails our loss test $t_{loss}$.

     \subsubsection{ The efficient classical Algorithm of Renema et al.}
     The efficient classical algorithm of Renema et al. \cite{RSG18}, which is a generalization of earlier work by a subset of the Authors \cite{RMC18} dealing only with distinguishability, takes into account both photon loss and distinguishability. For given values of $\lambda$ and $x$, the key idea behind the algorithm of \cite{RSG18} is that noisy BosonSampling of size $n$ can be viewed as an ideal BosonSampling with $k$  identical photons, supplemented by $n-k$ totally distinguishable particles. The classical algorithm of \cite{RSG18} takes as input a precision $\epsilon$, a size $n$ of a noisy boson sampler, $\lambda$ and $x$, and a confidence level $0<\delta<1$. This algorithm can sample from a probability distribution with $\mathsf{TVD} \leq \epsilon$ with respect to the distribution of  noisy BosonSampling of size $n$, with probability at least $1-\delta$ over Haar random $m \times m$ unitaries $U$. The complexity of the classical algorithm is $O(k2^kn^k)$, where $k$ is given by \cite{RSG18}
     \begin{equation}
     \label{eqk}
     k=\lceil 2\frac{log(\epsilon)+log(\delta)+log(1-\alpha)}{log(\alpha)} \rceil,
     \end{equation}
     where $\alpha=(1-\lambda)x^2.$ If $\lambda$ and $x$ are constant, then for a fixed $\epsilon$ and $\delta$, $k$ is also constant and the algorithm is therefore $poly(n)$-time. We will now suppose that $\epsilon$ is small enough so that the statistics of the efficient classical algorithm can be thought of as statistics coming from a noisy BosonSampling of size $n$ with $\lambda$ and $x$ constant.
     
     Since $\lambda=constant$, it will fail our loss test $t_{loss}$ after some value of $n$. Also, if $x=constant$, and for a fixed loss $l \approx \lceil \lambda n \rceil$ in the vicinity of the mean value (which appears with high probability in our experiments, see section \ref{secsamplcomp}), Equation (\ref{eqtheoravph2}) can be used to show that 
     \begin{multline*}
    %\label{eqtheoravph2}
    |\mathbf{E}_{U}(P_{K,l}(0_{K+1}\dots0_{m}))_{Th,id}-\mathbf{E}_{U}(P_{K,l}(0_{K+1}\dots0_{m}))_{Th,x}| \approx \\ (1-x^2)\frac{(n-l-1)(n-l)}{m}(1-O(\frac{(n-l-1)(m-K)}{m}) \\ \approx (1-x^2)\frac{(n-l-1)(n-l)}{m} \approx (1-x^2)(1-\lambda)^2\frac{n^2}{m}  \approx O(\frac{1}{n^{\gamma}})\gg o(\frac{1}{n^{\gamma}}).
    \end{multline*}
This means that  $x=constant$ (and $\lambda=constant$) fails our $t_{d_4}$ test, after some fixed value of $n$. In order to pass both $t_{loss}$ and $t_{d_4}$, we must have $x=1-o(1)=1-O(\frac{1}{n^{\beta}})$ and $\lambda=o(1)=O(\frac{1}{n^{\beta_2}})$ with $\beta, \beta_2 \in \mathbb{R}^{+*}$. This immediately implies (by substituting these values in Equation (\ref{eqk})) that $k$ will scale will $n$, and therefore the classical algorithm of \cite{RSG18} is no longer $poly(n)$-time. 

Note that a similar classical algorithm was given in \cite{Moylett} with a better run-time than that of Renema et al. \cite{RSG18}. However, the run-time of this algorithm still has an exponential dependence on $k$ \cite{Moylett}, and therefore is inefficient for our purposes by an argument similar to the one developed  in this section for Renema et al.'s algorithm. 
     
     \subsubsection{ The greedy sampler, and permanent approximation approaches inspired from Gurvits et al.}
     The greedy sampler was introduced in \cite{Google21}. Although used as a spoofing tool for Gaussian BosonSampling, it could in principle be generalized to our setting for BosonSampling, as remarked in \cite{Google21}. The idea behind this sampler is to produce $L$ $m$-bit strings with their $w$th order marginal distribution (i.e.\ the marginal distribution on $w$ modes, with $w \in \{1,..,m\}$) being $O(1/L)$ close to the $w$ th order distribution of an ideal boson sampler of size $n$ \cite{Google21}. The complexity of this algorithm is $O(m^w2^wL)$. The probability  $P_{K,l,U}(0_{K+1}\dots0_{m})$ at the heart of our $t_{d_4}$ test is a $w$th order marginal probability with $w=m-K=n-1$. Thus, the complexity of the greedy sampler algorithm needed to produce this probability with $1/poly(n)$ precision (to be able to pass $t_{d_4}$) is $O(m^n2^n/poly(n))$, which is exponential in $n$ and therefore not efficient.
     
     Similarly, one could think of using Gurvits' algorithm \cite{Gurvits2005,AH12} and try to brute force compute the probability $P_{K,l,U}(0_{K+1}\dots0_{m})$ by using the relation $$P_{K,l,U}(0_{K+1}\dots0_{m})=\sum_{s_1,\dots,s_K}P((s_1,\dots,s_K,0_{K+1},\dots0_{m})).$$
     Indeed, each element of the above sum is directly proportional to a permanent of an $n-l \times n-l$ matrix, and Gurvits' algorithm gives us a way of computing this permanent up to $1/poly(n-l)$ additive error in $poly(n-l)$-time. However, the number of terms of the above sum scales exponentially with $n-l$ for $m-K=n-1$. Therefore, the overall error on the approximation $P_{K,l,U}(0_{K+1}\dots0_{m})$ would be very high (exponential in $n-l$) and thus this approximation will fail $t_{d_4}$.
     
     Finally, it is worth noting that Gurvits \cite{AA11} gave another algorithm which \emph{exactly} computes the entire $w$th order marginal distribution in time $O((n-l)^{O(w)})$. However, in our case $w=n-1$, which makes this algorithm inefficient for producing $P_{K,l,U}(0_{K+1}\dots0_{m})$ which can pass $t_{d_4}$.
     \subsubsection{An attack based on leveraging known information}
     Because the expectation values over the Haar measure of our quantities of interest can be computed analytically exactly for  many types particles (identical photons, completely distinguishable particles,\dots) \cite{W16}, one could think of efficient classical adverserial strategies which use these analytically computed values to pass our developed tests. Here we develop one such strategy, designed to pass $t_{d_4}$, and show that this particular strategy will fail other tests such as $t_{d_1}$. For simplicity of discussion, we will assume $l=0$.
     
   %  Consider the following probability distribution over bit strings $s \in \mathcal{S}_{m,n}$(see section \ref{prelim}).
     
     Let $\mathcal{S}_1=\{s | s=\{s_1,\dots,s_K,0_{K+1},\dots0_{m}\}, s_1+\dots+s_K=n\}$, and $\mathcal{S}_2$ is such that $\mathcal{S}_1 \cap \mathcal{S}_2= \emptyset$, and $\mathcal{S}_1 \cup \mathcal{S}_2=\mathcal{S}_{m,n}$ (see section \ref{prelim}). Note that $|\mathcal{S}_{m,n}|={m \choose n}$, and $|\mathcal{S}_{1}|={K \choose n}$ \cite{AA11}. As before, $m-K=n-1$. Furthermore, let $$\alpha_n=\mathbf{E}_{U}(P_{K,0}(0_{K+1}\dots0_{m}))_{Th,id}=1-O(\frac{(n)(m-K)}{m}) \approx 1-O(\frac{1}{n^{\gamma}}).$$ Consider the following distribution
     \begin{equation}
     \label{eqDad}
     D_{ad}:=\{p(s_1)=\frac{\alpha_n}{|\mathcal{S}_1|}, p(s_2)=\frac{1-\alpha_n}{|\mathcal{S}_2|}  | s_1 \in \mathcal{S}_{1}, s_2 \in \mathcal{S}_2\}.
     \end{equation}
$D_{ad}$ can be sampled from efficiently classically, as its just a mixture of two uniform distributions.  An adverserial strategy where, for each given choice of $U$, bit strings are sampled from $D_{ad}$  can pass the test $t_{d_4}$. Indeed, computing
  \begin{multline*}
  \mathbf{E}_U(P_{K,0,ad}(0_{K+1}\dots0_{m}))=P_{K,0,ad}(0_{K+1}\dots0_{m})=\sum_{s_1 \in \mathcal{S}_1}p(s_1)=\alpha_n \\ =\mathbf{E}_{U}(P_{K,0}(0_{K+1}\dots0_{m}))_{Th,id}.
     \end{multline*}
     Thus,
     $$|\mathbf{E}_U(P_{K,0,ad}(0_{K+1}\dots0_{m}))-\mathbf{E}_{U}(P_{K,0}(0_{K+1}\dots0_{m}))_{Th,id}|=0<<o(\frac{1}{n^{\gamma}}),$$
     and therefore this adverserial strategy passes $t_{d_4}$. We will now show that this strategy fails the test $t_{d_1}$. But first, we prove the following theorem.
     \begin{theorem}
     \label{th1}
     $D_{ad}$ can be well approximated by the uniform distribution $D_{unif}:=\{p(s)=\frac{1}{|\mathcal{S}_{m,n}|}|s \in \mathcal{S}_{m,n}\}$ for large enough $n$, meaning that $\|D_{ad}-D_{unif}\|\leq O(\frac{1}{n^{\gamma}})$.
     \end{theorem}
     
     \bigskip
     \begin{proof}
  \begin{multline*}
  \|D_{ad}-D_{unif}\|=\frac{1}{2}\sum_{s} |p_{ad}(s)-p_{unif}(s)|=\frac{|\mathcal{S}_1|}{2}|\frac{\alpha_n}{|\mathcal{S}_1|}-\frac{1}{|\mathcal{S}_{m,n}|}|+\frac{|\mathcal{S}_2|}{2}|\frac{1-\alpha_n}{|\mathcal{S}_2|}-\frac{1}{|\mathcal{S}_{m,n}|}|.
     \end{multline*}
    For large $n$, $\alpha_{n} \approx 1$. Replacing this in the above expression and rearranging while noting that $|\mathcal{S}_2|=|\mathcal{S}_{m,n}|-|\mathcal{S}_1|$, we obtain  
    $$\|D_{ad}-D_{unif}\|\approx 1-\frac{|\mathcal{S}_1|}{|\mathcal{S}_{m,n}|}.$$
    
     Plugging in the fact that $m-K=n-1$, and $K=O(m)$, we get that
     $$\frac{|\mathcal{S}_1|}{|\mathcal{S}_{m,n}|}=\frac{{K \choose n}}{{K+n-1 \choose n}}>1-\frac{n^2}{K} \geq 1-O(\frac{1}{n^{\gamma}}).$$
     Where the last two terms in the above inequality follow from the bosonic birthday paradox bound \cite{AA11}, and from the fact that $K=O(m)=O(n^{2+\gamma}).$ Plugging this into the expression for $\mathsf{TVD}$, we get
     $$\|D_{ad}-D_{unif}\| \leq O(\frac{1}{n^{\gamma}}),$$ and the proof is complete.
      \end{proof}
     With Theorem \ref{th1} in hand, we will now use the (simpler to deal with) uniform distribution to compute what is needed for $t_{d_1}$. Let us firt compute $ \langle \mathbf{n}_i\mathbf{n}_j \rangle_{unif}$, the expectation value of the correlator $\langle \mathbf{n}_i\mathbf{n}_j \rangle $ under $D_{unif}$. Recall we are working in the no-collision regime, so 
     $$\langle \mathbf{n}_i\mathbf{n}_j \rangle_{unif}=p(1_i1_j).$$ where $p(1_i1_j)$  means the probability of getting one boson in mode $i$, and one boson in mode $j$. %$$p(1_i0_j)=\sum_{s_1=1,s_2=0,s_3,\dots,s_m}\frac{1}{{m \choose n}}=\frac{{m-2 \choose n-1}}{{m \choose n}}.$$ Similarly, $p(0_i1_j)=p(1_i0_j)$ and
     
     $p(1_i1_j)=\frac{{m-2 \choose n-2}}{{m \choose n}} \approx \frac{1}{n^{2+2\gamma}}+(-1+\frac{1}{n^{1+\gamma}})\frac{1}{n^{3+2\gamma}}$ Thus,$$\langle \mathbf{n}_i\mathbf{n}_j \rangle_{unif} \approx \frac{1}{n^{2+2\gamma}}+(-1+\frac{1}{n^{1+\gamma}})\frac{1}{n^{3+2\gamma}} .$$
    By a similar calculation,
     $ \langle\mathbf{n}_i \rangle_{unif} \langle\mathbf{n}_j \rangle_{unif}=(p(1_i))^2 =(\frac{{m-1 \choose n-1}}{{m \choose n}} )^2 =\frac{1}{n^{2+2\gamma}}.$ Therefore,
     $$\mathbf{E}(C_{ij})=C_{ij,unif}=\mathbf{E}_{U}(C_{ij,unif})= \langle \mathbf{n}_i\mathbf{n}_j \rangle_{unif}-\langle \mathbf{n}_i \rangle_{unif} \langle\mathbf{n}_j\rangle_{unif} \approx (-1+\frac{1}{n^{1+\gamma}})\frac{1}{n^{3+2\gamma}} $$
     Finally, $NM_{unif}=\frac{m^2}{n}C_{ij,unif} \approx -1+\frac{1}{n^{1+\gamma}}$. By choosing the precision of the $t_{d_1}$ test to be $o(\frac{1}{n})=\frac{1}{n^{1+\epsilon_2}}$ with $\epsilon_2>\gamma$, and noting that $NM_{Th,id} \approx -1-o(\frac{1}{n})=-1-\frac{1}{n^{1+\gamma}}$ \cite{W16}; we get that $|NM_{unif}-NM_{Th,id}|=O(\frac{1}{n^{1+\gamma}})\gg \frac{1}{n^{1+\epsilon_2}}$, and therefore this adverserial strategy will fail the $t_{d_1}$ test.
     
     \bigskip
     As concluding remarks for this section, we stress that we do not rule out the existance of more sophisticated adverserial strategies capable of passing all our developed tests, and indeed this was not our goal here. Our goal was to show that the known efficient classical simulation strategies for noisy BosonSampling \emph{experiments} \cite{RSG18,OB18,ONF21} (characterized by a certain fixed photon loss and distinguishability) cannot pass our tests indefinitely, and we have obtained as a bonus that many adverserial strategies for spoofing BosonSampling cannot produce statistics which pass all our tests.

     \subsection{Passing $t_{loss}$ and $t_{d_4}$ is \emph{nessesary} for any BosonSampling experiment claiming quantum computational speedup}
     
     %\section{Sufficient noise levels to pass PQF tests}
    % \label{experiment}
     In this section, we will provide further evidence that PQF is a good metric for characterizing the performance of photonic NISQ devices, by showing that passing the tests $t_{loss}$ and $t_{d_4}$ (which are used in computing PQF as seen before) is \emph{nessesary} for \emph{any} BosonSampling experiment claiming a quantum computational speedup, in a sense we will now specify precisely. Consider the following definition of an efficient classical algorithm for (weakly) simulating BosonSampling.
     
     \begin{definition}
     \label{def2}
     Let $C$ be a classical algorithm which, for a given linear optical circuit $U$ (an $m \times m$ unitary) can sample from a probability distribution $D_{C,U} := \{p_{c,u}(s)|s \in \mathcal{S}_{m,n}\}$. Let $0<\epsilon<1$ and $0<\delta<1$ be \textbf{fixed} numbers. Furthermore, let $\tilde{D}_U=\{\tilde{p}(s)|s \in \mathcal{S}_{m,n}\}$ be the probability distribution sampled from a noisy boson sampler of size $n$, with a linear optical circuit $U$. 
     %with average photon loss $\lambda$ and average distinguishability $x$. Assume that assumption $(A5)$  holds (see section \ref{secnoise}).
     We say that $C$  \textbf{efficiently} weakly simulates this boson sampler if $\|\tilde{D}_U-D_{C,U}\|\leq \epsilon$ for at least a $1-\delta$ fraction of Haar random  $m \times m$ unitaries $U$, and $C$ is $poly(n)$-time.
     \end{definition}
     Note that this definition of classical simulability of BosonSampling differs from that of \cite{AA11} in two ways. The first being that in our definition $\epsilon$ is fixed, whereas in \cite{AA11} this $\epsilon$ is a variable, and $C$ is $poly(n,\frac{1}{\epsilon})$-time. The second being that we allow the algorithm to fail for some fraction of Haar random unitaries, parametrized by $\delta$. Also note that definitions of weak classical simulability similar to ours were used to claim quantum computational speedup for other families of sampling problems such as random quantum circuits and IQP circuits \cite{MH17}.

 We  say that experiments carried out with a noisy boson sampler of size $n$  admit a \emph{quantum computational speedup} if, for some fixed $\epsilon$ and $\delta$ in Definition \ref{def2}, \textbf{no} $poly(n)$-time classical algorithm $C$ exists which can efficiently weakly simulate this boson sampler, in the sense of Definition \ref{def2}.
 
For  given fixed values of $\epsilon$ and $\delta$, from the results of \cite{RSG18}, it can directly be seen that an algorithm $C$ exists which efficiently weakly simulates, in the sense of Definition \ref{def2}, a noisy boson sampler of size $n$ with $x=constant$ and $\lambda=constant$, and even for the case when $\alpha=x^2(1-\lambda) \to 0$  as $n \to \infty$ (i.e.\ $\lambda \to 1$ and/or $x \to 0$, as in this case $k \to 0$ in Equation (\ref{eqk}), see section \ref{tests} and \cite{RSG18}). Thus, any BosonSampling experiment hoping to demonstrate a quantum  computational speedup must necessarily have $\lambda=o(1)=O(\frac{1}{n^{\beta}}) \to 0$ and/or $x=1-o(1)=1-O(\frac{1}{n^{\beta_2}}) \to 1$, where $\beta, \beta_2 \in \mathbb{R}^{+*}$. It can immediately be seen that this value of $\lambda$ will pass $t_{loss}$ (see the definition of $t_{loss}$ in section \ref{tests}). For $x$, using Equation (\ref{eqtheoravph2}) for $l=\lceil \lambda n \rceil$ (and $\lambda=constant$ or $\lambda \to 0$), we get
     \begin{multline*}
    %\label{eqtheoravph2}
    |\mathbf{E}_{U}(P_{K,l}(0_{K+1}\dots0_{m}))_{Th,id}-\mathbf{E}_{U}(P_{K,l}(0_{K+1}\dots0_{m}))_{Th,x}| \approx \\ (1-x^2)\frac{(n-l-1)(n-l)}{m}(1-O(\frac{(n-l-1)(m-K)}{m}) \\ \approx (1-x^2)\frac{(n-l-1)(n-l)}{m} \approx O(\frac{n^2}{mn^{\beta_2}})  \approx O(\frac{1}{n^{\gamma+\beta_2}})<<o(\frac{1}{n^{\gamma}}),
    \end{multline*}
which means that this value of $x$ passes $t_{d_{4}}$.

To summarize, any BosonSampling experiment demonstrating a quantum computational speedup must nessesarily pass either (or both) of $t_{loss}$ and $t_{d_4}$.

      \section{Sufficient noise levels to pass PQF tests}
    \label{experiment}
    
    A remaining question to answer is: what should the values of $\lambda$ and $x$ be for a noisy boson sampler of size $n$ to keep on passing our tests ? We answer this question by proving the following Theorem.
    \begin{theorem}
     \label{thexp}
     A noisy boson Sampler of size $n$ with $\lambda=o(1)$ and $x=1-o(\frac{1}{n^6})$  produces output statistics which pass  the tests $t_{loss}$ and $t_{d_1}$ to $t_{d_4}$ for all $n \geq n_0$, where $n_0 \in \mathbb{N}^{*}$.
    \end{theorem}
    We will prove Theorem \ref{thexp} in appendix \ref{secproof}, but we will briefly discuss the main technical tools involved in the proof here. Proving that the tests $t_{loss}$ and $t_{d_4}$ are passed by such a boson sampler is straightforward, by using the definition of $t_{loss}$ and Equation (\ref{eqtheoravph2}). To prove that this boson samplers statistics pass $t_{d_1}$ to $t_{d_3}$ is more involved. To do this, we use the expansion developed  in \cite{RSG18,RMC18} for the permanent of  a noisy boson sampler, as well as the upper bounds for the coefficients of this expansion averaged over the Haar measure of $m \times m$ unitaries \cite{RMC18}. We use these to compute bounds on the values of $NM$,$CV$, $S$ in the case where the boson sampler has a distinguishability $x$, which is a generalization of the computations of these quantities carried out in \cite{W16} for the ideal case ($x=1$) and which are mentioned in section \ref{PQF}.
    
    To conclude this section, we note that the values of $x$ and $\lambda$ needed to pass all our tests in Theorem \ref{thexp} give a $\textrm{PQF}_{\gamma,\{\epsilon\}}=+\infty$ for some $\gamma$ and $\{\epsilon\}$ (see section \ref{PQF}). Although these values may seem very stringent from an experimental point of view, we note that these are analytically computed values based on approximations and upper bounds. It might be that numerical explorations of this problem can give much more practical values of $x$
     and $\lambda$ that pass all our tests. We leave such explorations for future work.
     
     %%Note u dont need to change proof, just say u work in vicinity of mean and note that $(1-\lambda)n \approx n$, meaning we can as if just work with n particles in the proof\dots
     \section{Discussion}
     \label{discussion}
     To summarize, we have introduced a single number metric, the Photonic Quality Factor (PQF), and presented evidence that it is a reliable metric for assessing the average performance of a noisy photonic quantum device based on single-photon sources, linear optical circuits, and single-photon detectors, in which the main sources of noise are photon loss and distinguishability.
     Several interesting questions and directions present themselves, the most immediate being that other sources of noise can also be considered.
     
     Numerical explorations could lead to less stringent bounds on photon loss ($\lambda$) and distinguishability ($x$) for passing the tests, which may lead to more experimentally-friendly targets. The idea of well-motivated weakening of the requirements for the tests used to evaluate PQF could also be pursued, with the goal of deriving more easily attainable yet still meaningful values of $\lambda$ and $x$. 
     
     We can also consider how the metric could be made more applicable beyond photonic quantum computing for other hardware, similar to benchmarks developed in \cite{B1,B2,B3,qscore}. A first step in this direction could be the work of \cite{Moylett2}, where it is shown how to simulate BosonSampling using a quantum circuit composed of quantum gates. Since our tests are designed to assess the quality of BosonSampling experiments, it may be that simulating BosonSampling using the technique of \cite{Moylett2}, then performing our developed tests on the output statistics can give us an indication about the quality of the set of gates used in this simulation. However, further investigation is required to be able to make concrete claims, and thus we leave this as an avenue for future investigation.
     
     Our results, similar to those of \cite{RSG18,RaulGeom2020,RMC18}, highlight the fact that some error correction and mitigation techniques must be introduced after some value of $n$, otherwise the boson sampler of size $n$ would become efficiently classically simulable. This is manifested in the fact that $x$ must tend to one (identical photons), and $\lambda$ to zero (lossless regime) as $n$ gets larger, to keep on passing our developed tests (see Theorem \ref{thexp}). In this direction, it would be interesting to work out what value of PQF is needed to perform a useful computation using a photonic NISQ device, where no error correction is yet available \cite{NISQ}. This could give us an indication of whether something useful could be done with these devices, or whether we would have to wait for error corrected versions of these devices.
     
     Finally, a future direction we will pursue is generalizing the noise models introduced here to include more sources of error, such as dark counts of the  single-photon detector for example \cite{S2020}, or perhaps including some time dependence in the noise, as well as some dependence of the noise on the geometry of the linear optical circuits \cite{RaulGeom2020}, all with the goal of developing more sophisticated techniques to benchmark realistic photonic devices where a large set of errors come into play in non-trivial ways.
     \section*{Acknowledgements}
     We thank Frederic Grosshans, Niccolo Somaschi, Nicolas Maring, and Andreas Fyrillas for fruitful discussions.

\onecolumn\newpage
\appendix

\section{Proof of Theorem \ref{thexp}}
\label{secproof}
We first begin by noting that $\lambda=o(1)$ will pass $t_{loss}$, from the definition of this test (see section \ref{sectloss}). We will now turn to the tests $t_{d_1}$
to $t_{d_4}$. For large enough $n \geq n_0$, $n_{rem}=n-\tilde{l} \approx n-\lceil \lambda n \rceil \approx n$, where $\tilde{l}$  is the mean number of lost photons, and $n_{rem}$ the mean number of remaining photons. Thus, we will take $n_{rem} \approx n$; meaning that we will prove Theorem \ref{thexp} for the case of $n$ single photons with distinguishability $x \in [0,1]$.

For $t_{d_4}$, by using Equation (\ref{eqtheoravph2}) for $x=1-o(\frac{1}{n^6})$
\begin{multline*}
    \mathbf{E}_{U}(P_{K,0}(0_{K+1}\dots0_{m}))_{Th,id}-\mathbf{E}_{U}(P_{K,0}(0_{K+1}\dots0_{m}))_{Th,x}=\\ (1-x^2)\frac{(n-1)(n)}{m}\mathbf{E}_{U}(P_{K,1}(0_{K+1}\dots0_{m}))_{Th,id} \approx O(\frac{n^2}{n^6m}) \\ \approx O(\frac{1}{n^{6+\gamma}})<<o(\frac{1}{n^{\gamma}}).
\end{multline*}
Thus, this value of $x$ will pass $t_{d_4}$. 

Now we turn to the tests $t_{d_1}$ to $t_{d_3}$. Before we proceed we will prove the following Lemma.
\begin{lemma}
\label{lem}
For $\beta:=1-x=o(\frac{1}{n^2})$, we  have that
\begin{multline*}
  \frac{1-x^{n+1}}{1-x}=n+1-\frac{n(n+1)}{2}(1-x) +\kappa(x)  \approx n+1- \frac{n(n+1)}{2}(1-x), 
\end{multline*}
for $n \geq n_0$, and where $|\kappa(x)|=o(\frac{n(n+1)}{2}(1-x))$.
\end{lemma}

\begin{proof}
$\frac{1-x^{n+1}}{1-x}$ can be written as the following geometric series
\begin{equation}
\label{eqlem1}
\frac{1-x^{n+1}}{1-x}=\sum_{i=0,..,n}x^{i}=1+\sum_{i=1,\dots,n}(1-\beta)^{i}.
\end{equation}
We can write
\begin{equation}
    \label{eqlem2}
    \sum_{i=1,\dots,n}(1-\beta)^{i}=\sum_{i=1,..,n}\sum_{j=0,..,i}(-\beta)^{i-j}{i \choose j}.
\end{equation}
performing the relabelling $i-j:=k$ in Equation (\ref{eqlem2}) and replacing this is Equation (\ref{eqlem1}) we obtain
\begin{equation}
  \label{eqlem3}
  \frac{1-x^{n+1}}{1-x}=1+\sum_{i=1,..,n}\sum_{k=0,..,i}(-\beta)^{k}{i \choose i-k}.
\end{equation}
   Equation (\ref{eqlem3}) can be rewritten as

\begin{multline}
  \label{eqlem4}
  \frac{1-x^{n+1}}{1-x}=1+(-\beta)^{0}\sum_{j=1,\dots,n}{j \choose j}-\beta \sum_{j=1,\dots,n}{j \choose j-1}+ \sum_{k=2,..,n}(-\beta)^{k}\sum_{j=k,\dots,n}{j \choose j-k} \\= n+1-\frac{n(n+1)}{2}\beta+\sum_{k=2,..,n}(-\beta)^{k}\sum_{j=k,\dots,n}{j \choose j-k}=n+1-\frac{n(n+1)}{2}\beta+\kappa(x),
\end{multline}
where $\kappa(x):=\sum_{k=2,..,n}(-\beta)^{k}\sum_{j=k,\dots,n}{j \choose j-k}$.

Now, 
$$(\beta)^{k}\sum_{j=k,\dots,n}{j \choose j-k} = (\beta)^{k} \sum_{j=k,\dots,n} \frac{j!} {(j-k)!k!} \leq (\beta)^{k} \sum_{j=k,\dots,n} j^k \leq (\beta)^{k} \sum_{j=k,\dots,n} n^k \leq O((\beta)^{k}n^{k+1}). $$
%Replacing this in Equation (\ref{eqlem4}) we get
%\begin{multline}
 For $\beta=o(\frac{1}{n^2})=O(\frac{1}{n^{2+\epsilon}})$, we have $O((\beta)^{k}n^{k+1})=O(\frac{1}{n^{k+\epsilon k-1}})=o(\frac{1}{n^{1+\epsilon}})$
for $k>1$. Thus,
$$|\kappa(x)|=|\sum_{k=2,..,n}(-\beta)^{k}\sum_{j=k,\dots,n}{j \choose j-k}| \leq (n-1)o(\frac{1}{n^{1+\epsilon}})=o(\frac{1}{n^{\epsilon}})=o(\frac{n(n+1)}{2}\beta)=o(\frac{n(n+1)}{2}(1-x))$$
This completes the proof of Lemma \ref{lem}.
\end{proof}

We will now start by computing $\mathbf{E}_U(C_{ij})_x$, for a given distinguishability $x \in [0,1]$ of the $n$ single photons (see section \ref{tests}). Henceforth, we will use the approximation ${m \choose n} \approx \frac{m^n}{n!}$ (note that this was also used in \cite{Moylett}). For $s=\{s_1,\dots,s_m\}$ and $\sum_{i=1,..m}s_i=n$, we will use the following expansion of $P(s)$, the probability to observe the output $s$ of the boson sampler, derived in \cite{RMC18}
$$P(s)=\sum_{i=0,..,n}C_{i,s} x^{i},$$
 $C_{i,s}$ are expansion coefficients satisfying \cite{RMC18}
%$$\mathbf{E}_U(|C_{0,s}|)= \frac{n!}{m^n},$$
%$$\mathbf{E}_U(|C_{1,s}|)=0,$$

$$\mathbf{E}_U(|C_{i,s}|) \leq \frac{n!}{m^n}.$$

\begin{multline*}
    \mathbf{E}_U(\langle \mathbf{n}_i\mathbf{n}_j \rangle)_x=\sum_{s:=\{s_1,\dots,s_{i-1},1_i,s_{i+1},..,s_{j-1},1_j,s_{j+1},\dots,s_m\}}\mathbf{E}_U(p(s)) \\ = \sum_{s:=\{s_1,\dots,s_{i-1},1_i,s_{i+1},..,s_{j-1},1_j,s_{j+1},\dots,s_m\}} \sum_{i=0,..,n}\mathbf{E}_U(C_{i,s})x^{i} \\ \leq \sum_{s:=\{s_1,\dots,s_{i-1},1_i,s_{i+1},..,s_{j-1},1_j,s_{j+1},\dots,s_m\}} \sum_{i=0,..,n}\mathbf{E}_U(|C_{i,s}|)x^{i} \\ \leq \sum_{s:=\{s_1,\dots,s_{i-1},1_i,s_{i+1},..,s_{j-1},1_j,s_{j+1},\dots,s_m\}} \sum_{i=0,..,n}\frac{n!}{m^n}x^{i} \\ \leq \sum_{s:=\{s_1,\dots,s_{i-1},1_i,s_{i+1},..,s_{j-1},1_j,s_{j+1},\dots,s_m\}} \frac{n!}{m^n} \frac{1-x^{n+1}}{1-x} \\ \leq {m-2 \choose n-2 } \frac{n!}{m^n} \frac{1-x^{n+1}}{1-x}.
\end{multline*}
Note  that ${m-2 \choose n-2} \approx \frac{n^2}{m^2} {m \choose n} \approx \frac{1}{n^{2+2\gamma}} {m \choose n} \approx \frac{1}{n^{2+2\gamma}} \frac{m^n}{n!}. $
Replacing this in the above equation gives

\begin{equation}
  \label{eqninj}
  \mathbf{E}_U(\langle \mathbf{n}_i\mathbf{n}_j \rangle)_x \leq \frac{1}{n^{2+2\gamma}}\frac{1-x^{n+1}}{1-x}.
\end{equation}
A similar calculation for $\mathbf{E}_U(\langle \mathbf{n}_i \rangle \langle \mathbf{n}_j\rangle)_x=\mathbf{E}_{U}(p(1_i)p(1_j))$ gives
\begin{equation}
    \label{eqnitimesnj}
    \mathbf{E}_U(\langle \mathbf{n}_i\rangle \langle \mathbf{n}_j\rangle)_x \leq \frac{1}{n^{2+2\gamma}}(\frac{1-x^{n+1}}{1-x})^2,
\end{equation}
where $p(1_i)=\sum_{s:=\{s_1,\dots,s_{i-1},1_i,s_{i+1},\dots,s_m\}}p(s)$ is the probability to have one photon in mode $i$ (similarly for $p(1_j)$).

Now, we compute 

\begin{multline}
    \label{eqcijapp}
    |\mathbf{E}_U(C_{ij})_{x=1}-\mathbf{E}_U(C_{ij})_{x}| \leq |\mathbf{E}_U(\langle \mathbf{n}_i\mathbf{n}_j\rangle )_{x=1}-\mathbf{E}_U(\langle \mathbf{n}_i\mathbf{n}_j\rangle_{x})|+\\ |\mathbf{E}_U(\langle \mathbf{n}_i\rangle \langle \mathbf{n}_j \rangle_{x=1})-\mathbf{E}_U(\langle \mathbf{n}_i\rangle \langle \mathbf{n}_j \rangle_{x})|.
\end{multline}

\begin{multline}
\label{eqfirstpart}
    |\mathbf{E}_U(\langle \mathbf{n}_i\mathbf{n}_j\rangle _{x=1})-\mathbf{E}_U(\langle \mathbf{n}_i\mathbf{n}_j\rangle_{x})|=|\sum_{s:=\{s_1,\dots,s_{i-1},1_i,s_{i+1},..,s_{j-1},1_j,s_{j+1},\dots,s_m\}} \sum_{i=0,..,n}\mathbf{E}_U(C_{i,s})(1-x^{i})| \leq \\ \sum_{s:=\{s_1,\dots,s_{i-1},1_i,s_{i+1},..,s_{j-1},1_j,s_{j+1},\dots,s_m\}} \sum_{i=0,..,n}\mathbf{E}_U(|C_{i,s}|)(1-x^{i}) \leq \\ \sum_{s:=\{s_1,\dots,s_{i-1},1_i,s_{i+1},..,s_{j-1},1_j,s_{j+1},\dots,s_m\}} \sum_{i=0,..,n}\frac{n!}{m^n}(1-x^{i}) \\ \leq  \frac{1}{n^{2+2\gamma}}(n+1-\frac{1-x^{n+1}}{1-x}).
\end{multline}
Where the last part of this Equation is obtained by a similar calculation to that in Equation (\ref{eqninj}). Also by a similar calculation we get
\begin{equation}
    \label{eqsecpart}
    |\mathbf{E}_U(\langle \mathbf{n}_i \rangle \langle \mathbf{n}_j \rangle_{x=1})-\mathbf{E}_U(\langle \mathbf{n}_i \rangle \langle \mathbf{n}_j\rangle_{x})| \leq \frac{1}{n^{2+2\gamma}}((n+1)^2-(\frac{1-x^{n+1}}{1-x})^2)
\end{equation}
Replacing these in Equation (\ref{eqcijapp}) gives
\begin{multline}
    \label{eqcijapp2}
    |\mathbf{E}_U(C_{ij})_{x=1}-\mathbf{E}_U(C_{ij})_{x}| \leq \frac{1}{n^{2+2\gamma}}(n+1-\frac{1-x^{n+1}}{1-x}) + \frac{1}{n^{2+2\gamma}}((n+1)^2-(\frac{1-x^{n+1}}{1-x})^2).
\end{multline}
The above Equation directly implies (from the monotonicity of the expectation value)
\begin{multline}
    \label{eqcijapp3}
    |\mathbf{E}(\mathbf{E}_U(C_{ij})_{x=1})-\mathbf{E}(\mathbf{E}_U(C_{ij})_{x})| \leq \frac{1}{n^{2+2\gamma}}(n+1-\frac{1-x^{n+1}}{1-x}) + \frac{1}{n^{2+2\gamma}}((n+1)^2-(\frac{1-x^{n+1}}{1-x})^2).
\end{multline}
Here we use $\mathbf{E}(.)$ to denote the average over the $C$-data set which contains ${m \choose 2} \approx m^2$ terms $C_{ij}$ (see section \ref{tests}).
Multiplying both sides of Equation (\ref{eqcijapp3}) by $\frac{m^2}{n}$, we obtain
\begin{multline}
\label{eqnm1}
 |NM_{Th,id}-NM_{x}| \leq \frac{m^2}{n}( \frac{1}{n^{2+2\gamma}}(n+1-\frac{1-x^{n+1}}{1-x}) + \frac{1}{n^{2+2\gamma}}((n+1)^2-(\frac{1-x^{n+1}}{1-x})^2)).
\end{multline}
Using Lemma \ref{lem} and replacing $m=n^{2+2\gamma}$ and $x=1-O(\frac{1}{n^{2+\epsilon}})$, we get after a long but straightforward calculation that 
$$|NM_{Th,id}-NM_{x}| \leq O(\frac{1}{n^{-2+\epsilon}}).$$
In order to pass $t_{d_1}$, we must have $-2+\epsilon>1$, that is, $\epsilon=3+\delta$, with $\delta \in \mathbb{R}^{+*}$. So far, we have passed $t_{d_1}$ with an $x$ that looks like $$x=1-O(\frac{1}{n^{5+\delta}})=1-o(\frac{1}{n^5}).$$ We still however need to pass $t_{d_2}$ and $t_{d_3}$.

Plugging $x$ into Equation (\ref{eqcijapp2}) and using Lemma \ref{lem}, we can directly see that 
$$\mathbf{E}_U(C_{ij})_x=\mathbf{E}_U(C_{ij})_{x=1}+\epsilon_{ij},$$
where $$|\epsilon_{ij}| \leq O(\frac{1}{n^{4+2\gamma+\delta}}).$$
Similarly,

$$\mathbf{E}(\mathbf{E}_U(C_{ij}))_x=\mathbf{E}(\mathbf{E}_U(C_{ij}))_{x=1}+\tilde{\epsilon},$$
where $$|\tilde{\epsilon}| \leq \sum_{i,j}\frac{|\epsilon_{ij}|}{m^2} \leq O(\frac{1}{n^{4+2\gamma+\delta}}) <<|\mathbf{E}(\mathbf{E}_U(C_{ij})_{x=1})|=\frac{n}{m^2}|NM_{Th,id}|=\frac{1}{n^{3+2\gamma}}.$$
Where we have used $NM_{Th,id} \approx -1$ (see section \ref{PQF}).

In a slight approximation which will ease calculation, we will take $$\epsilon_{ij}-\tilde{\epsilon} \approx \tilde{\varepsilon},$$ where $|\tilde{\varepsilon}| \leq O(\frac{1}{n^{4+2\gamma+\delta}})$, and where we will assume $\tilde{\varepsilon}$ is independent of $i,j$.

With this new notation in hand, we now go on to evaluating $CV$ for the $t_{d_2}$ test.
\begin{multline}
\label{eqcv1}
|CV^2_{x=1}-CV^2_{x}| \approx |\frac{\frac{1}{m^2}\sum_{ij}(\mathbf{E}_U(C_{ij})_{x=1}-\mathbf{E}(\mathbf{E}_U(C_{ij}))_{x=1})^2-(\mathbf{E}_U(C_{ij})_{x}-\mathbf{E}(\mathbf{E}_U(C_{ij}))_{x})^2}{(\mathbf{E}(\mathbf{E}_U(C_{ij}))_{x=1})^2}|.
\end{multline}
We have used the approximation $(\mathbf{E}(\mathbf{E}_U(C_{ij}))_{x=1})^2\approx (\mathbf{E}(\mathbf{E}_U(C_{ij}))_{x})^2$ to write the above expression using a single denominator. This approximation is valid since, as seen before,  $|\tilde{\epsilon}|<<|\mathbf{E}(\mathbf{E}_U(C_{ij}))_{x=1}|$. 

Plugging in the above defined $\epsilon_{ij}$, $\tilde{\epsilon}$, and $\tilde{\varepsilon}$ in Equation (\ref{eqcv1}) we get
\begin{multline}
\label{eqcv2}
|CV^2_{x=1}-CV^2_{x}| \approx \\ |\frac{\frac{1}{m^2}\sum_{ij}(\mathbf{E}_U(C_{ij})_{x=1}-\mathbf{E}(\mathbf{E}_U(C_{ij}))_{x=1})^2-(\mathbf{E}_U(C_{ij})_{x=1}-\mathbf{E}(\mathbf{E}_U(C_{ij}))_{x=1}-\tilde{\varepsilon})^2}{\frac{n^2}{m^4}}| \leq \\ \approx \frac{m^4}{n^2}\tilde{\varepsilon}^2.
\end{multline}
Where the rightmost part of this equation is obtained by expanding $(\mathbf{E}_U(C_{ij})_{x=1}-\mathbf{E}(\mathbf{E}_U(C_{ij}))_{x=1}-\tilde{\varepsilon})^2$ while noting that $\sum_{ij}\mathbf{E}_U(C_{ij})_{x=1}-\mathbf{E}(\mathbf{E}_U(C_{ij}))_{x=1}=m^2\mathbf{E}(\mathbf{E}_U(C_{ij}))_{x=1}-m^2\mathbf{E}(\mathbf{E}_U(C_{ij}))_{x=1}=0$.
%$\tilde{$
Now, $$\frac{m^4}{n^2}\tilde{\varepsilon}^2 \leq O(\frac{1}{n^{2+2\delta}}) <<o(\frac{1}{n}),$$ and therefore, 
we pass the test $t_{d_2}$ for $x=1-O(\frac{1}{n^{5+\delta}})=1-o(\frac{1}{n^5}).$

Finally, for Skewness

\begin{multline}
\label{eqs1}
S_x = \frac{\sum_{ij}(\mathbf{E}_U(C_{ij})_{x=1}-\mathbf{E}(\mathbf{E}_U(C_{ij})_{x=1})-\tilde{\varepsilon})^3}{(\sum_{ij}(\mathbf{E}_U(C_{ij})_{x=1}-\mathbf{E}(\mathbf{E}_U(C_{ij})_{x=1})-\tilde{\varepsilon})^2)^{\frac{3}{2}}} \\=\frac{\sum_{ij}(\mathbf{E}_U(C_{ij})_{x=1}-\mathbf{E}(\mathbf{E}_U(C_{ij})_{x=1}))^3-f_1(n)}{(\sum_{ij}(\mathbf{E}_U(C_{ij})_{x=1}-\mathbf{E}(\mathbf{E}_U(C_{ij})_{x=1}))^2+f_2(n))^{\frac{3}{2}}}.
\end{multline}
with $f_1(n)=3\tilde{\varepsilon}\sum_{ij}(\mathbf{E}_U(C_{ij})_{x=1}-\mathbf{E}(\mathbf{E}_U(C_{ij})_{x=1}))^2+m^2\tilde{\varepsilon}^3,$ and $f_2(n)=m^2\tilde{\varepsilon}^2.$ The rightmost side of Equation (\ref{eqs1}) is obtained by a direct expansion of the numerator and denominator.

In order for $|S_x-S_{x=1}|=o(\frac{1}{n})$, and therefore for us to pass $t_{d_3}$, we need first that 
$$|f_2(n)|<<\sum_{ij}(\mathbf{E}_U(C_{ij})_{x=1}-\mathbf{E}(\mathbf{E}_U(C_{ij})_{x=1}))^2$$ so that we can write the above difference with a common denominator of $\sum_{ij}(\mathbf{E}_U(C_{ij})_{x=1}-\mathbf{E}(\mathbf{E}_U(C_{ij})_{x=1}))^2:=\sigma^2$.
From the definition of $CV$, we have $\sigma^2=(\mathbf{E}(\mathbf{E}_U(C_{ij}))_{x=1})^2.CV_{Th,id}=\frac{n^2}{m^4}(NM_{Th,id}CV_{Th,id})^2 \approx \frac{n^2}{m^4}$ (see section \ref{PQF} for values of $NM_{Th,id}$ and $CV_{Th,id}$). Therefore, we need
$$|\tilde{\varepsilon}^2| << \frac{n^2}{m^6},$$ or 

$$O(\frac{1}{n^{8+4\gamma+2\delta}}) <<  \frac{1}{n^{10+6\gamma}}.$$
If we choose $\delta>1+\gamma=1+\gamma+\delta^{'},$ where $\delta^{'} \in \mathbb{R}^{+*}$, we get that the above inequality is verified. Note that this makes our value of $x$ to be $$x=1-o(\frac{1}{n^6})=1-O(\frac{1}{n^{6+\gamma+ \delta^{'}}}).$$ Now, we can write 
$$|S_{x}-S_{x=1}| \approx \frac{|f_1(n)|}{(\sigma^2)^{1.5}} \approx \frac{m^6}{n^3}|f_{1}(n)|.$$ Expanding 
$$\frac{m^6}{n^3}|f_{1}(n)| \approx 3\tilde{\varepsilon}\frac{m^6}{n^3}\frac{n^2}{m^4}+\frac{m^8}{n^3} \tilde{\varepsilon}^3 \approx O(\frac{1}{n^{2+\gamma+\delta^{'}}})+O(\frac{1}{n^{2+\gamma+3\delta^{'}}})<<o(\frac{1}{n}),$$
meaning that we pass $t_{d_3}$.

\bigskip 
To summarize, if $x=1-o(\frac{1}{n^6})$, this value of $x$ can pass all tests $t_{d_1}$ to $t_{d_3}$. This completes the proof of Theorem \ref{thexp}.

\section{The routing circuit}
\label{approuting}
Here we show a linear optical circuit for transforming the Fock state $|s_1,\dots,s_{m}\rangle$
with $s_{i} \in \{0,1\}$, $\sum_{i}s_i=n-l$, and $s_{n+1},\dots,s_{m}=0$ onto the \emph{canonical} state $|1_1,\dots,1_{n-l},0_{n-l+1},\dots0_{m}\rangle$
where the $n-l$ photons occupy the first $n-l$ modes. Because modes $n+1,\dots,m$ are never populated, we will instead just focus on the first 
$n$ modes. Our goal is to find a unitary $U_{route}$
which transforms $|s_1,\dots,s_n\rangle$ with $s_i \in \{0,1\}$ and $\sum_{i=1,\dots,n}s_i=n-l$ onto $|1_1\dots1_{n-l}0_{n-l+1}\dots0_n\rangle$.

One strategy to construct $U_{route}$ would be to perform a series of swaps of the single photons to bring them into the desired positions (note that these swaps are \emph{different} from the SWAP gate applied on a qubit/path encoded photon state). Recall that a general unitary two-mode linear optical transformation may be written as \cite{Kok07}
\begin{equation}
    \label{eqappb1}
    \mathbf{a}^{\dagger}_{out}=cos(\theta)\mathbf{a}^{\dagger}_{in}+ie^{-i\phi}sin(\theta)\mathbf{b}^{\dagger}_{in},
\end{equation}
\begin{equation}
    \label{eqappb2}
   \mathbf{b}^{\dagger}_{out}=ie^{i\phi}sin(\theta)\mathbf{a}^{\dagger}_{in}+ cos(\theta)\mathbf{b}^{\dagger}_{in}.
\end{equation}
Where $(\mathbf{a}_{in},\mathbf{b}_{in})$ and $(\mathbf{a}_{out},\mathbf{b}_{out})$ are respectively the input and output modes of the two-mode transformation, $\theta$ characterizes the transmitivity/refelecively of this transformation, and $\phi$ characterizes its phase shift (see \cite{Kok07} for details).
Also note that a \emph{phase shifter} with angle $\gamma$ can realize the following transformation
\begin{equation}
  \label{eqappb3}
  \mathbf{a}^{\dagger}_{out}=e^{i\gamma}\mathbf{a}^{\dagger}_{in}
\end{equation}
Taking $\theta=\frac{\pi}{2}$ and $\phi=0$ In Equations (\ref{eqappb1}) and (\ref{eqappb2}) we get
$$\mathbf{a}^{\dagger}_{out}=i\mathbf{b}^{\dagger}_{in},$$ and
$$\mathbf{b}^{\dagger}_{out}=i\mathbf{a}^{\dagger}_{in}.$$
Adding a phase shifter of $\gamma=-\frac{\pi}{2}$ at the level of each output mode of the two-mode transformation gives
$$\mathbf{a}^{\dagger}_{out}=\mathbf{b}^{\dagger}_{in},$$ and
$$\mathbf{b}^{\dagger}_{out}=\mathbf{a}^{\dagger}_{in}.$$
The above equations mean that a photon in any one of the input modes of this two-mode transformation will be swapped. We will henceforth refer to this two-mode transformation as our \emph{swap gadget}, and it is represented in Figure \ref{figapp1}.
 \begin{figure}[h!]
\includegraphics[scale=0.5]{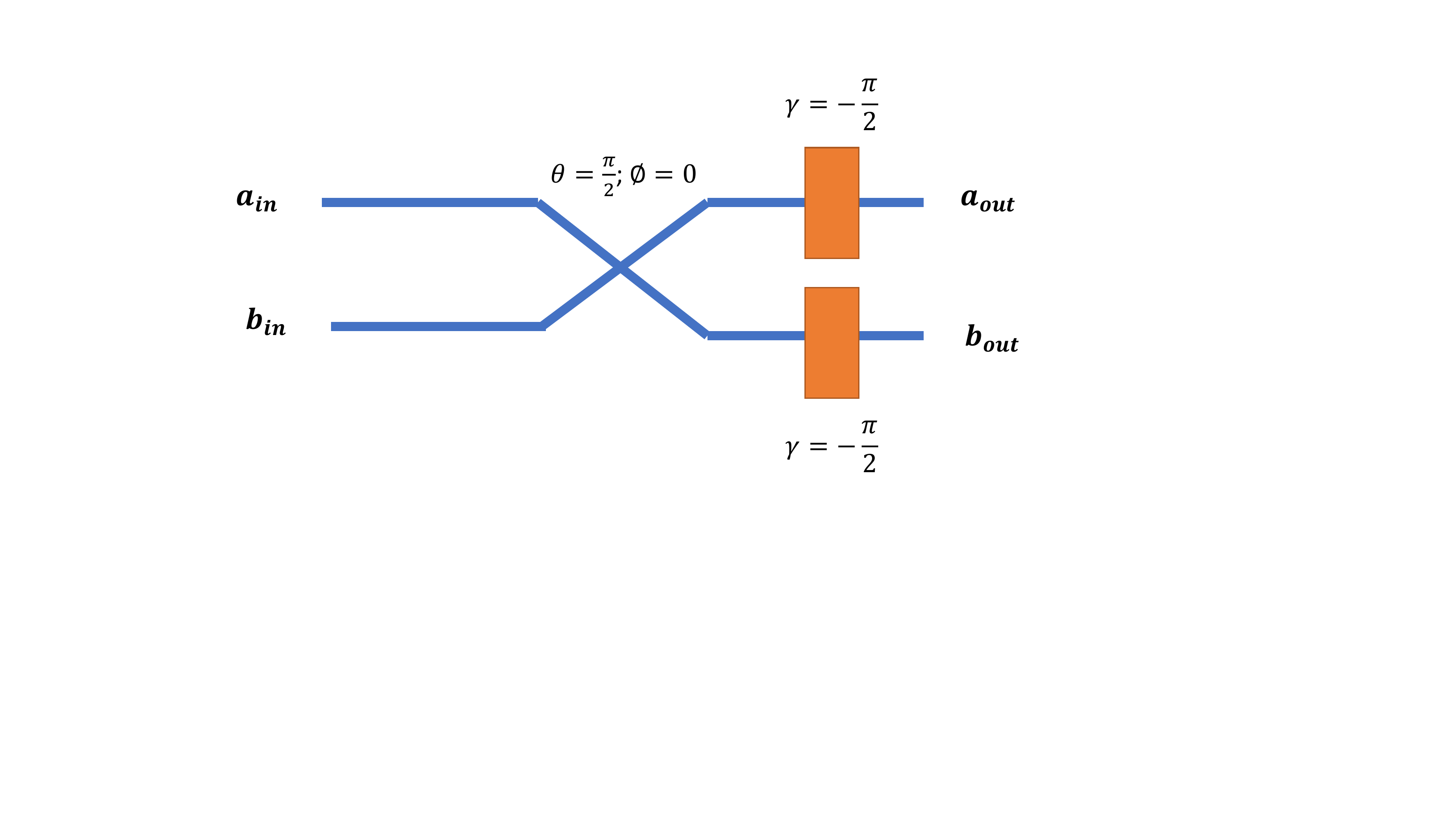}
%\centering
\caption{the swap gadget. the two-mode transformation is represented by the blue X shaped figure, above which are indicated the values of $\theta$ and $\phi$. The orange rectangles are the phase shifters, above which is indicated the value of $\gamma$. }
\label{figapp1}
\end{figure}
%\bigskip 
With the swap gadget in hand, we now describe our procedure for implementing $U_{route}$. %transforming $|s_1\dotss_n\rangle$ to $|1_1\dots1_{n-l}0_{n-l+1}\dots0_n\rangle$.

\begin{itemize}
    \item For mode number $i=1$, identify the closest mode to it which has occupancy one. Call this mode $j$. Then, swap the photon from mode $j$ to mode $i$. If $i \neq j$, and $i$ and $j$ are not adjacent (nearest-neighbor) modes; this swapping can be done using $O(j-i)$ swap gadgets applied on adjacent modes, starting from mode $j$  upwards onto mode $i$.
    \item Repeat this procedure for $i=2,\dots,n-l$.
\end{itemize}
An example of a linear optical circuit implementing this procedure is found in Figure \ref{fig2app}, transforming the Fock state $|0,0,1,1,1\rangle$ onto $|1,1,1,0,0\rangle$.
\begin{figure}[h!]
\includegraphics[scale=0.5]{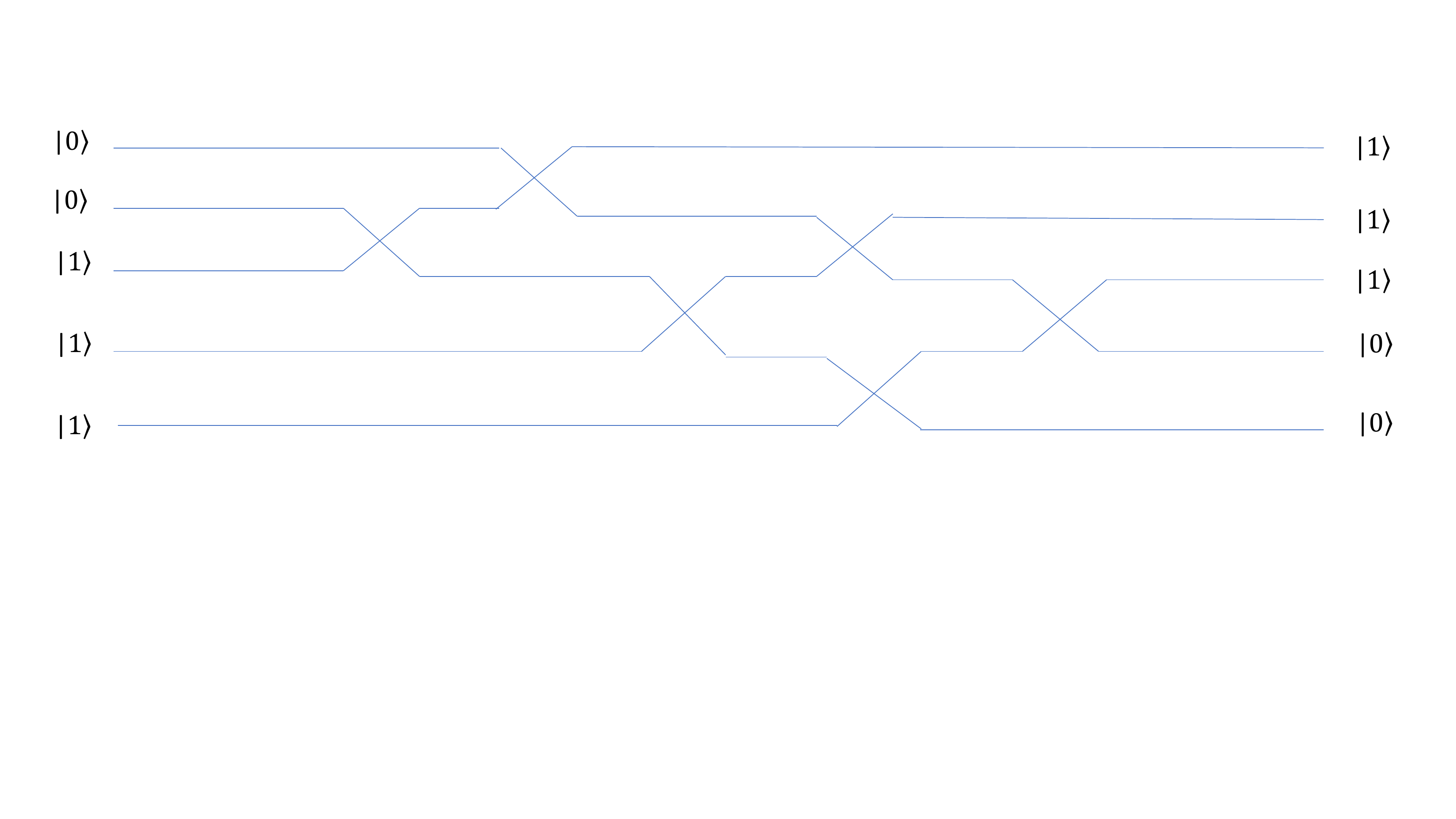}
%\centering
\caption{A linear optical circuit transforming the Fock state $|0,0,1,1,1\rangle$ onto its canonical form $|1,1,1,0,0\rangle$. The swap gadgets are the X shaped figures. $U_{route}$ here is a $5 \times 5$ unitary constructed from a product of all $2 \times 2$ unitaries of the all swap gadgets. }
\label{fig2app}
\end{figure}

\section{Concerning Equation (\ref{eqtheoravph2})}
\label{appC}
We will not re-derive the proof of Equation (\ref{eqtheoravph2}) in \cite{S16}, but rather we will show here that the conditions needed for the proof of (\ref{eqtheoravph2}) in \cite{S16} are satisfied by our model of distinguishability \cite{RMC18}. For a system of $n$
photons, the proof of \cite{S16} holds for single photon states of the form 
\begin{equation}
    \label{eqC1}
    \rho_i=|\phi\rangle\langle\phi|-\delta \rho_i,
\end{equation}
for $i=1,\dots,n$, $|\phi\rangle\langle\phi|^{\otimes n}$ is the state of an \emph{ideal} system of $n$ fully indistinguishable photons, and $\delta \rho_i$ is a perturbative term with $Tr(\delta \rho_i)=0$, and $\langle \phi| \delta \rho_i | \phi \rangle \approx F_{av}<<1$ for $i \in \{1,\dots,n\}$, where $F_{av}$ is an averaged fidelity \cite{S16}.

In our model for distinguishability, based on that of \cite{RMC18}, we have that two single photon states $|\phi_i\rangle$ and $|\phi_j\rangle$ satisfy
$\langle \phi_i| \phi_j \rangle=x$ for all $i \neq j \in \{1,..,n\}$. We can write

\begin{equation}
    \label{eqC2}
    |\phi_i\rangle=|\phi\rangle+(|\phi_i\rangle-|\phi\rangle),
\end{equation}
for $i \in \{1,\dots,n\}$.
Similarly, we can write
\begin{equation}
    \label{eqC3}
    \rho_i=|\phi_i\rangle \langle \phi_i|=|\phi \rangle \langle \phi|+(|\phi \rangle_i \langle \phi_i|-|\phi \rangle \langle \phi|).
\end{equation}
Calling $\delta \rho_i=|\phi \rangle \langle \phi|-|\phi_i\rangle \langle \phi_i|$, it is straightforward to see that $Tr(\delta \rho_i)=Tr(|\phi \rangle \langle \phi|-|\phi_i\rangle \langle \phi_i|)=Tr(|\phi \rangle \langle \phi|)-Tr(|\phi_i\rangle \langle \phi_i|)=0$. Furthermore, when $x \approx 1$, we have that $|\phi_i\rangle \langle \phi_i| \approx |\phi\rangle \langle \phi|$, and therefore, $\langle \phi | \delta \rho_i |\phi \rangle=\langle \phi |( |\phi \rangle \langle \phi|-|\phi_i\rangle \langle \phi_i|) |\phi \rangle=1-|\langle \phi_i| \phi \rangle |^2 <<1$. Thus, our chosen noise model encompasses all the conditions needed to prove Equation (\ref{eqtheoravph2}), and the proof of \cite{S16} follows straightforwardly.

\end{document}